%% file: multicut-in-dags.tex
\documentclass{article}
\usepackage{fullpage}

\usepackage[all=normal,bibliography=tight,bibnotes=tight]{savetrees}

\usepackage{amssymb,comment}
\usepackage{amsmath}
\usepackage{xspace}
\usepackage{tikz}
\usepackage{amstext}
\usepackage{calc}
\usepackage{amsopn}
\usepackage{latexsym}
\usepackage{tabularx}

\usepackage{amsthm}

\newtheorem{theorem}{Theorem}[section]

\newtheorem{lemma}[theorem]{Lemma}

\theoremstyle{definition}
\newtheorem{definition}[theorem]{Definition}

\newtheorem{reduction}{Reduction rule}
\newtheorem{branching}[reduction]{Branching rule}

\newcommand{\defproblem}[4]{
  \vspace{1mm}
\noindent\fbox{
  \begin{minipage}{0.95\textwidth}
  \begin{tabular*}{\textwidth}{@{\extracolsep{\fill}}lr} #1 & {\bf{Parameter:}} #3 \\ \end{tabular*}
  {\bf{Input:}} #2  \\
  {\bf{Question:}} #4
  \end{minipage}
  }
  \vspace{1mm}
}

\newcommand{\multicut}{{\textsc{Multicut}}}
\newcommand{\dagmulticut}{\multicut{} {\textsc{in DAGs}}}
\newcommand{\Oh}{O}
\newcommand{\Ohstar}{\ensuremath{O^\ast}}
\newcommand{\terms}{\ensuremath{\mathcal{T}}}
\newcommand{\nterms}{r}
\newcommand{\cut}{p}
\newcommand{\toporder}{\leq_\tau}
\newcommand{\sources}{T^s}
\newcommand{\sinks}{T^t}
\newcommand{\src}{S}

\newcommand{\reverse}[1]{#1^\textrm{rev}}
\newcommand{\shadow}{A}
\newcommand{\shfam}{\mathcal{A}}

\newcommand{\poten}{\phi}
\newcommand{\outN}{N^+}
\newcommand{\inN}{N^-}

\newcommand{\Pp}{\ensuremath{\mathcal{P}}}
\newcommand{\inst}{\mathcal{I}}

\begin{document}
\date{}
  \title{Fixed-parameter tractability of multicut in directed acyclic graphs}

\author{
  Stefan Kratsch\thanks{Utrecht University, Utrecht, the Netherlands, \texttt{s.kratsch@uu.nl}.} \and
	Marcin Pilipczuk\thanks{Institute of Informatics, University of Warsaw, Poland, \texttt{malcin@mimuw.edu.pl}.} \and
  Micha\l{} Pilipczuk\thanks{Department of Informatics, University of Bergen, Norway, \texttt{michal.pilipczuk@ii.uib.no}.} \and
  Magnus Wahlstr\"{o}m\thanks{Max-Planck-Institute for Informatics, Saarbr\"{u}cken, Germany, \texttt{wahl@mpi-inf.mpg.de}.}
}

\maketitle

\begin{abstract}
The \multicut{} problem, given a graph $G$, a set of terminal pairs $\terms=\{(s_i,t_i)\ |\ 1\leq i\leq r\}$ and an integer $p$, asks whether one can find a {\emph{cutset}} consisting of at most $p$ non-terminal vertices that separates all the terminal pairs, i.e., after removing the cutset, $t_i$ is not reachable from $s_i$ for each $1\leq i\leq r$. The fixed-parameter tractability of \multicut{} in undirected graphs, parameterized by the size of the cutset only, has been recently proven by Marx and Razgon \cite{marx:multicut} and, independently, by Bousquet~et~al.~\cite{thomasse:multicut}, after resisting attacks as a long-standing open problem. In this paper we prove that \multicut{} is fixed-parameter tractable on directed acyclic graphs, when parameterized both by the size of the cutset and the number of terminal pairs. We complement this result by showing that this is implausible for parameterization by the size of the cutset only, as this version of the problem remains $W[1]$-hard.
\end{abstract}

\input{intro}

\input{algorithm}

\input{lower-bounds}

\input{conclusions}

\bibliographystyle{plain}
\bibliography{multicut-in-dags}

\end{document}

%% file: intro.tex
\section{Introduction}

Parameterized complexity is an approach for tackling hard problems by designing algorithms that perform robustly, when the input instance is in some sense simple; its difficulty is measured by an integer that is additionally appended to the input, called the {\emph{parameter}}. Formally, we say that a problem is {\emph{fixed-parameter tractable}} (FPT), if it can be solved by an algorithm that runs in time $f(k)n^c$ for $n$ being the length of the input and $k$ being the parameter, where $f$ is some computable function and $c$ is a constant independent of the parameter.

The search for fixed-parameter algorithms resulted in the introduction of a number of new algorithmic techniques, and gave fresh insight into the structure of many classes of problems. One family that received a lot of attention recently is the so-called {\emph{graph cut}} problems, where the goal is to make the graph satisfy a global separation requirement by deleting as few edges or vertices as possible (depending on the variant). Graph cut problems in the context of fixed-parameter tractability were to our knowledge first introduced explicitly in the seminal work of Marx \cite{marx-prehistory}, where it was proved that (i) {\textsc{Multiway Cut}} (separate all terminals from each other by a cutset of size at most $p$) in undirected graphs is FPT when parameterized by the size of the cutset; (ii) {\textsc{Multicut}} in undirected graphs is FPT when parameterized by both the size of the cutset and the number of terminal pairs. 
Fixed-parameter tractability of {\textsc{Multicut}} parameterized by the size of the cutset only was left open by Marx~\cite{marx-prehistory}; resolved much later (see below).

The probably most fruitful contribution of the work of Marx \cite{marx-prehistory} is the concept of {\emph{important separators}}, which proved to be a tool almost perfectly suited to capturing the bounded-in-parameter character of sensible cutsets. The technique proved to be extremely robust and serves as the key ingredient in a number of FPT algorithms~\cite{chen-nmc,directed-fvs,dir-mwc,sfvs,Guillemot11a,clustering-daniels,marx-prehistory,marx:multicut,razgon:icalp2008}.
In particular, the fixed-parameter tractability of {\textsc{Skew Multicut}} in directed acyclic graphs, obtained via a simple application of important separators, enabled the first FPT algorithm for {\textsc{Directed Feedback Vertex Set}} \cite{directed-fvs}, resolving another long-standing open problem.

However, important separators have a drawback in that not all graph cut problems admit solutions with ``sensible'' cutsets in the required sense. This is particularly true in directed graphs, where, 
with the exception of the aforementioned {\textsc{Skew Multicut}} problem in {\textsc{DAGs}},
for a long time few fixed-parameter tractable graph cut problems were known;
in fact, up until very recently it was open whether {\textsc{Multiway Cut}} in directed graphs admits an FPT algorithm even in the restricted case of two terminals. 
The same complication arises in the undirected {\textsc{Multicut}} problem parameterized by the size of the cutset. 

After a long struggle, {\textsc{Multicut}} was shown to be FPT by Marx and Razgon~\cite{marx:multicut} and, independently, by Bousquet et al. \cite{thomasse:multicut}.
The key component in the algorithm of Marx and Razgon~\cite{marx:multicut} is the technique of \emph{shadow removal},
which, in some sense, serves to make the solutions to cut problems more well-behaved.
This was adapted to the directed case by Chitnis et al. \cite{dir-mwc}, who proved that {\textsc{Multiway Cut}}, parameterized by the size of the cutset, is fixed-parameter tractable for an arbitrary number of terminals, by a simple and elegant application of the shadow removal technique.
This gives hope that, in general, shadow removal may be helpful for the application of important separators to the directed world.

As for the directed {\textsc{Multicut}} problem, it was shown by Marx and Razgon~\cite{marx:multicut} 
to be $W[1]$-hard when parameterized only by the size of the cutset, but otherwise had unknown status, even for a constant number of terminals in a \textsc{DAG}. 

\paragraph{Our results.} The main result of this paper is the proof of fixed-parameter tractability of the \dagmulticut{} problem, formally defined as follows:

\defproblem{\dagmulticut}{Directed acyclic graph $G$, set of terminal pairs $\terms = \{(s_i,t_i)\ |\ 1 \leq i \leq \nterms\}$,
  $s_i,t_i \in V(G)$ for $1 \leq i \leq \nterms$, and an integer $\cut$.}{$p+r$}{Does there exist a set $Z$ of at most $\cut$ non-terminal vertices
    of $G$, such that for any $1 \leq i \leq \nterms$ the terminal $t_i$ is not reachable from $s_i$ in $G \setminus Z$?}

\begin{theorem}\label{thm:main}
\dagmulticut{} can be solved in $\Oh^*(2^{\Oh(\nterms^2\cut+\nterms 2^{\Oh(\cut)})})$ time.
\end{theorem}
Note, that throughout the paper we use~$\Oh^*$-notation to suppress polynomial factors.
Note also that we focus on vertex cuts; it is well known that in the directed acyclic
setting the arc- and vertex-deletion variants are~equivalent~(cf.~\cite{dir-mwc}). 

Our algorithm makes use of the shadow removal technique introduced by Marx and Razgon \cite{marx:multicut}, adjusted to the directed setting by Chitnis et al. \cite{dir-mwc}, as well as the basic important separators toolbox that can be found in \cite{dir-mwc}. We remark that the shadow removal is but one of a number of ingredients of our approach: in essence, the algorithm combines the shadow removal technique with a degree reduction for the sources in order to carefully prepare the structure of the instance for a simplifying branching step. A more detailed overview of a single step of the algorithm is depicted in Figure \ref{fig:diagram}, given in the appendix.

We complement the main result with two lower bounds.
First, we show that the dependency on $\nterms$ in the exponent is probably unavoidable.
\begin{theorem}\label{thm:lb-w1}
\dagmulticut{}, parameterized by the size of the cutset only, is $W[1]$-hard.
\end{theorem}
Thus, we complete the picture of parameterized complexity of \dagmulticut. We hope that it is a step towards fully understanding the parameterized complexity of \multicut{} in general directed graphs.

Second, we establish NP-completeness of \textsc{Skew Multicut}, a special case of \dagmulticut{} where we are given $d$ sources $(s_i)_{i=1}^d$ and $d$ sinks $(t_i)_{i=1}^d$, and the set of terminal pairs is defined as $\terms =\{ (s_i,t_j): 1 \leq i \leq j \leq d\}$. Recall that the FPT algorithm for \textsc{Skew Multicut} is the core subroutine of the algorithm for \textsc{Directed Feedback Vertex Set} of Chen et al.~\cite{directed-fvs}.
NP-completeness of \dagmulticut{} with two terminal pairs is an easy corollary of Theorem~\ref{thm:lb-skew}.
\begin{theorem}\label{thm:lb-skew}
\textsc{Skew Multicut} is NP-complete even in the restricted case of two sinks and two sources.
\end{theorem}
\begin{theorem}\label{thm:lb-np}
\dagmulticut{} is NP-complete even in the restricted case of two terminal pairs.
\end{theorem}

\paragraph{Organization of the paper.}
In Section \ref{sec:prel} we introduce notation and recall the notion of important
separators and the technique of shadow removal of Marx and Razgon \cite{marx:multicut}
and Chitnis et al. \cite{dir-mwc}.
Section \ref{sec:alg} contains the proof of our main contribution, Theorem \ref{thm:main}.
The lower bounds (i.e., Theorems~\ref{thm:lb-w1}, \ref{thm:lb-skew} and~\ref{thm:lb-np})
  are proven in Section \ref{sec:lb}.
Section \ref{sec:conc} concludes the paper.

%% file: algorithm.tex
\section{Preliminaries}\label{sec:prel}

For a directed graph $G$, by $V(G)$ and $E(G)$ we denote its vertex- and arc-set, respectively.
For a vertex $v \in V(G)$, we define its {\em{in-neighbourhood}} $\inN_G(v) = \{u: (u,v) \in E(G)\}$
and {\em{out-neighbourhood}} $\outN_G(v) = \{u: (v,u) \in E(G)\}$;
these definitions are extended to sets $X \subseteq V(G)$ by
$\inN_G(X) = ( \bigcup_{v \in X} \inN_G(v) ) \setminus X$ and
$\outN_G(X) = ( \bigcup_{v \in X} \outN_G(v) ) \setminus X$.
The {\em{in-degree}} and {\em{out-degree}} of $v$ are defined as $|\inN_G(v)|$ and $|\outN_G(v)|$, respectively.
In this paper we consider simple directed graphs only; if at any point a modification of the graph results in a multiple arc,
we delete all copies of the arc except for one. By $\reverse{G}$ we denote the graph $G$ with all the arcs reversed, i.e., $\reverse{G}=(V(G),\{(v,u):(u,v)\in E(G)\})$.

A {\em{path}} in $G$ is a sequence of pairwise different vertices $P=(v_1,v_2,\ldots,v_d)$ such that
$(v_i,v_{i+1}) \in E(G)$ for any $1 \leq i < d$. If $v_1$ is the first vertex of the path $P$ and $v_d$ is the last vertex,
  we say that $P$ is a {\em{$v_1v_d$-path}}. We extend this notion to sets of vertices: if $v_1 \in X$ and $v_d \in Y$ for some
$X,Y \subseteq V(G)$, then $P$ is a $XY$-path as well.
For a path $P=(v_1,v_2,v_3,\ldots,v_d)$ the vertices $v_2,v_3,\ldots,v_{d-1}$ are the {\em{internal vertices}} of $P$. The set of internal vertices of a path $P$ is the {\em{interior}} of $P$.
We say that a vertex $v$ is {\em{reachable}} from a vertex $u$ in $G$ if there exists a $uv$-path in $G$.
As the considered digraphs are simple, each path $P=(v_1,v_2,\ldots,v_d)$ has a unique \emph{first arc} $(v_1,v_2)$ and a unique \emph{last arc} $(v_{d-1},v_d)$.

Let $(G,\terms,\cut)$ be a \dagmulticut{} instance
with a set of $\nterms$ terminal pairs $\terms = \{(s_i,t_i): 1 \leq i \leq \nterms\}$.
We call the terminals $s_i$ {\em{source terminals}} and the terminals $t_i$ {\em{sink terminals}}. We let $\sources = \{s_i: 1 \leq i \leq \nterms\}$,
$\sinks = \{t_i: 1 \leq i \leq \nterms\}$ and $T = \sources \cup \sinks$.
Consider the following easy reduction.

\begin{lemma}\label{lem:source-sink}
There exists a polynomial-time algorithm that, given a \dagmulticut{} instance $(G,\terms,\cut)$,
computes and equivalent instance $(G',\terms',\cut')$, such that:
\begin{enumerate}
\item $|\terms| = |\terms'|$ and $\cut = \cut'$;
\item $\terms' = \{(s_i',t_i'): 1 \leq i \leq \nterms\}$ and all terminals $s_i'$ and $t_i'$ are pairwise different;
\item for each $1 \leq i \leq \nterms$ we have $\inN_{G'}(s_i') = 0$ and $\outN_{G'}(t_i') = 0$.
\end{enumerate}
\end{lemma}
\begin{proof}
To construct graph $G'$, start with the graph $G$ and for each terminal $v$ of $\terms$ replace $v$ with $\cut+1$ copies of $v$ (i.e., vertices with the same in- and out-neighbourhood as $v$;
    these copies are not terminals in $\terms'$).
Moreover, for each $1 \leq i \leq \nterms$, add to $G'$ a new terminal $s_i'$ with arcs $\{(s_i',u): u \in \outN_G(s_i)\}$
and a new terminal $t_i'$ with arcs $\{(u,t_i'): u \in \inN_G(t_i)\}$. Finally, take $\terms'=\{(s_i',t_i'): 1 \leq i \leq \nterms\}$
and $\cut'=\cut$.

To see the equivalence, first take a multicut $Z$ in $(G,\terms,\cut)$ and an arbitrary $s_i't_i'$-path $P' = (s_i',v_1',v_2',\ldots,v_d',t_i')$ in $G'$.
The path $P'$ induces a path $P = (s_i,v_1,v_2,\ldots,v_d,t_i)$ in $G$: $v_k = v_k'$ if $v_k'$ is present in $G$, and $v_k = v$ if $v_k'$ is one of the $\cut+1$
copies of a terminal $v$ in $\terms$. The path $P$ is intersected by $Z$ on some non-terminal vertex $v_k$; as $v_k=v_k'$, $Z$ intersects $P'$. We infer that $Z$ is a multicut
in $(G',\terms',\cut')$ as well.

In the other direction, let $Z$ be a multicut in $(G',\terms',\cut')$ of size at most $\cut$ and let $P = (s_i,v_1,v_2,\ldots,v_d,t_i)$ be an arbitrary $s_it_i$-path in $G$.
Construct a path $P' = (s_i',v_1',v_2',\ldots,v_d',t_i')$ as follows: take $v_k'=v_k$ if $v_k$ is a nonterminal in $(G,\terms,\cut)$ and take $v_k'$ to be one of the copy
of $v_k$ that is not contained in $Z$ otherwise; note that at least one copy is not contained in $Z$, as $|Z|\leq p$. As $Z$ is a multicut in $(G',\terms',\cut')$, $Z$ intersects $P'$. By the choice of the internal vertices of $P'$,
$Z$ intersects $P'$ on some vertex $v_k'=v_k$ for $v_k$ being a nonterminal in $(G,\terms,\cut)$. Therefore $Z \cap V(G)$ is a multicut in $(G,\terms,\cut)$.
\end{proof}
In our algorithm, the set of terminal pairs $\terms$ is never modified, and neither in-neighbors of a source nor out-neighbors of a sink are added.
Thus, we may assume that during the course of our algorithm
all terminals are pairwise distinct and that
$\inN_G(s_i) = \outN_G(t_i) = \emptyset$ for all $1 \leq i \leq \nterms$.

Fix a topological order $\toporder$ of $G$. 
For any sets $X,Y \subseteq V(G)$, we may order
the vertices of $X$ and $Y$ with respect to $\toporder$, and compare $X$ and $Y$ lexicographically;
we refer to this order on subsets of $V(G)$ as the {\em{lexicographical order}}.

A set $Z \subseteq V(G)$ is called a {\em{multicut}} in $(G,\terms,\cut)$, if $Z$ contains no terminals,
  but for each $1 \leq i \leq \nterms$, $t_i$ is not reachable from $s_i$ in $G \setminus Z$.
Given a \dagmulticut{} instance $\inst = (G,\terms,\cut)$ a multicut $Z$ is called a {\em{solution}} if $|Z| \leq \cut$.
A solution $Z$ is called a {\em{lex-min solution}} if $Z$ is lexicographically minimum solution in $\inst$ among solutions of minimum possible size.

For $v \in V(G)$, by $\src(G,v)$ we denote set of source terminals $s_i$ for which there exists a $s_iv$-path in $G$.
For a set $S \subseteq \sources$ by $V(G,S)$ we denote the set of nonterminal vertices $v$ for which $\src(G,v) = S$.

\subsection{Important separators and shadows}

In the rest of this section we recall the notion of important separators by Marx~\cite{marx-prehistory}, adjusted to the directed case by Chitnis et al. \cite{dir-mwc},
   as well as the shadow removal technique of Marx and Razgon \cite{marx:multicut}
   and Chitnis et al. \cite{dir-mwc}.

\begin{definition}[\textbf{separator}, \cite{dir-mwc}, Definition 2.2]
Let $G$ be a directed graph with terminals $T \subseteq V(G)$.
Given two disjoint non-empty sets $X,Y \subseteq V(G)$, we call a set $Z \subseteq V(G)$ an {\em{$X-Y$ separator}}
if (i) $Z \cap T = \emptyset$, (ii) $Z \cap (X \cup Y) = \emptyset$, (iii) there is no path from $X$ to $Y$ in $G \setminus Z$.
An $X-Y$ separator $Z$ is called {\em{minimal}} if no proper subset of $Z$ is a $X-Y$ separator.
\end{definition}

By $\textrm{cut}_G(X,Y)$ we denote the size of a minimum $X-Y$ separator in $G$; $\textrm{cut}_G(X,Y) = \infty$
if $G$ contains an arc going directly from $X$ to $Y$.
By Menger's theorem, $\textrm{cut}_G(X,Y)$ equals the maximum possible size of a family of $XY$-paths with pairwise disjoint interiors.

\begin{definition}[\textbf{important separator}, \cite{dir-mwc}, Definition 4.1]
Let $G$ be a directed graph with terminals $T \subseteq V(G)$ and let $X,Y \subseteq V(G)$  be two disjoint non-empty sets.
Let $Z$ and $Z'$ be two $X-Y$ separators. We say that $Z'$ is {\em{behind}} $Z$ if any vertex reachable from $X$ in $G \setminus Z$
is also reachable from $X$ in $G \setminus Z'$. A minimal $X-Y$ separator is an {\em{important separator}} if
no other $X-Y$ separator $Z'$ satisfies $|Z'| \leq |Z|$ while being also behind $Z$.
\end{definition}

We need also some known properties of minimum size cuts (cf. \cite{chen-nmc,dir-mwc}).

\begin{lemma}[\cite{dir-mwc}, Lemma B.4]
Let $G$ be a directed graph with terminals $T \subseteq V(G)$.
For two disjoint non-empty sets $X,Y \subseteq V(G)$, there exists exactly one minimum size important $X-Y$ separator.
\end{lemma}

\begin{definition}[\textbf{closest mincut}]
Let $G$ be a directed graph with terminals $T \subseteq V(G)$.
For two disjoint non-empty sets $X,Y \subseteq V(G)$, the
unique minimum size important $X-Y$ separator is called the {\em{$X-Y$ mincut closest to $Y$}}.
The {\em{$X-Y$ mincut closest to $X$}} is the $Y-X$ mincut closest to $X$ in $\reverse{G}$.
\end{definition}

\begin{lemma}\label{lem:push-mincut}
Let $G$ be a directed graph with terminals $T \subseteq V(G)$
and let $X,Y \subseteq V(G)$ be two disjoint non-empty sets.
Let $B$ be the unique minimum size important $X-Y$ separator, that is, the $X-Y$ mincut
closest to $Y$, and let $v \in B$ be an arbitrary vertex.
Construct a graph $G'$ from $G$ as follows: delete $v$ from $G$ and add an arc $(x,w)$
for each $x \in X$ and $w \in \outN_G(v)\setminus X'$, where $X'$ is the set of vertices reachable from $X$ in $G\setminus B$.
Then the size of any $X-Y$ separator in $G'$ is strictly larger than $|B|$.
\end{lemma}

\begin{proof}
The claim is obvious if $Y \cap \outN_G(v) \neq \emptyset$, as then $G'$ contains
a direct arc from $X$ to $Y$. Therefore, let us assume that $Y \cap \outN_G(v) = \emptyset$.

We prove the lemma by exhibiting more than $|B|$ $XY$-paths in $G'$ that have pairwise
disjoint interiors.
Recall that $X' \supseteq X$ is the set of vertices reachable from $X$ in $G \setminus B$;
note that $\outN_G(X') = B$. Since $B$ is a minimum size $X-Y$ separator,
there exist a set of $XY$-paths $(P_u)_{u \in B}$ such that
$P_u$ intersects $B$ only in $u$ and the interiors of paths $P_u$
are pairwise disjoint. Note that each path $P_u$ can be split into two parts:
$P_u^X$, between $X$ and $u$, with all vertices except for $u$ contained in $X'$,
  and $P_u^Y$, between $u$ and $Y$, with all vertices except for $u$ contained
  in $Y' = V(G) \setminus (X' \cup B)$.

Consider a graph $G''$ defined as follows: we take the graph $G'[V(G') \setminus (X'\cup \{v\})]$
and add a terminal $s$ and arcs $(s,w)$ for all $w \in \outN_{G'}(X') = (B \setminus \{v\}) \cup
(\outN_G(v)\setminus X')$. Let $B'$ be a minimum size $s-Y$ separator in $G''$. 

We claim that $B'$
is an $X-Y$ separator in $G$. Let $P$ be an arbitrary $XY$-path in $G$ and let $u$ be the last (closest to $Y$)
  vertex of $B$ on $P$. Then in $G''$ there exists a shortened version $P'$ of $P$:
  if $u = v$ and $w$ is the vertex directly after $u$ on $P$, then $P'$ starts
  with the arc $(s,w)$ and then follows $P$ to $Y$ (observe that $w\notin X'$, as $v$ was the last vertex of $B$ on $P$), while
  if $u \neq v$, then $P'$ starts with the arc $(s,u)$ and then follows $P$ to $Y$. As $B'$ has to intersect $P'$, then $B'$ also intersects $P$.

Moreover, we claim that $B'$ is behind $B$. This follows directly from the fact that $G''$ does not contain $X'$, so $X'$ is still reachable from $X$ in $G\setminus B'$.

As $B$ is an important separator, $B'$ is behind $B$ and $B' \neq B$,
   we have that $|B'| > |B|$. Therefore, there exists a family $\Pp$ of at least $|B|+1$ 
   $sY$-paths in $G''$ that have pairwise disjoint interiors. 
Observe that all these paths are disjoint with $X'$ by the construction of $G''$.
 Each path from $\Pp$ that starts with an arc $(s,w)$ for $w \in \outN_G(v)\setminus X'$
 is present (with the first vertex replaced by an arbitrary vertex of $X$)
  in $G'$ as well.
 Moreover, each other path starts with an arc $(s,u)$ for $u \in B$;
 in $G'$ such a path can be concatenated with the $Xu$-path $P_u^X$.
 All paths $P_u^X$ are entirely contained in $X'$ except for the endpoint $u$,
 so we obtain the desired family of $XY$-paths in $G'$.
\end{proof}

We now recall the necessary definitions of the shadow
removal technique from \cite{dir-mwc}.

\begin{definition}[\textbf{shadow}, \cite{dir-mwc}, Definition 2.3]
Let $G$ be a directed graph and $T \subseteq V(G)$ be a set of terminals. Let $Z \subseteq V(G)$ be a subset of vertices.
Then for $v \in V(G)$ we say that
\begin{enumerate}
\item $v$ is in the {\em{forward shadow}} of $Z$ (with respect to $T$), if $Z$ is a $T-v$ separator in $G$, and
\item $v$ is in the {\em{reverse shadow}} of $Z$ (with respect to $T$), if $Z$ is a $v-T$ separator in $G$.
\end{enumerate}
\end{definition}

\begin{definition}[\textbf{thin}, \cite{dir-mwc}, Definition 4.4]
Let $G$ be a directed graph and $T \subseteq V(G)$ a set of terminals.
We say that a set $Z \subseteq V(G)$ is {\em{thin}} in $G$ if there is no $v \in Z$
such that $v$ belongs to the reverse shadow of $Z \setminus \{v\}$ with respect to $T$.
\end{definition}

\begin{theorem}[\textbf{derandomized random sampling}, \cite{dir-mwc}, Theorem 4.1 and Section 4.3]\label{marx:sampling}
There is an algorithm that, given a directed graph $G$, a set of terminals $T \subseteq V(G)$
and an integer $\cut$, produces in time $\Ohstar(2^{2^{\Oh(\cut)}})$ a family $\shfam$ of size
$2^{2^{\Oh(\cut)}} \log |V(G)|$ of subsets of $V(G) \setminus T$ such that the following holds.
Let $Z \subseteq V(G) \setminus T$ be a thin set with $|Z| \leq p$ and let $Y$
be a set such that for every $v \in Y$ there is an important $v-T$ separator $Z_v \subseteq Z$.
For every such pair $(Z,Y)$ there exists a set $\shadow \in \shfam$ such that $\shadow \cap Z = \emptyset$
but $Y \subseteq \shadow$.
\end{theorem}

We use Theorem \ref{marx:sampling} to identify vertices
separated from all sources in a given \dagmulticut{} instance.

\begin{definition}[\textbf{source shadow}]
Let $(G,\terms,\cut)$ be a \dagmulticut{} instance  and $Z \subseteq V(G)$ be a subset of nonterminals in $G$.
We say that $v \in V(G)$ is in {\em{source shadow}} of $Z$ if $Z$ is a $\sources-v$ separator.
\end{definition}

\begin{lemma}[\textbf{derandomized random sampling for source shadows}]\label{lem:shadowless}
There is an algorithm that, given a \dagmulticut{} instance $(G,\terms,\cut)$, 
produces in time $\Ohstar(2^{2^{\Oh(\cut)}})$ a family $\shfam$ of size $2^{2^{\Oh(\cut)}} \log |V(G)|$ of subsets of nonterminals of $G$ such that
if $(G,\terms,\cut)$ is a YES instance and $Z$ is the lex-min solution to $(G,\terms,\cut)$,
then there exists $\shadow \in \shfam$ such that $\shadow \cap Z = \emptyset$ and all vertices of source shadows of $Z$ in $G$ are contained in $\shadow$.
\end{lemma}

\begin{proof}
Let $Y$ be the set of vertices of source shadows of $Z$. To proof the lemma it is sufficient to apply Theorem \ref{marx:sampling} for the graph $\reverse{G}$
with terminals $\sources$ and budget $\cut$. Thus, we need to prove that the pair $(Z,Y)$ satisfies properties given in Theorem \ref{marx:sampling}.

First, assume that $Z$ is not thin in $\reverse{G}$ w.r.t. $\sources$. Let $v \in Z$ be a witness: $Z \setminus \{v\}$ is a $v-\sources$ separator in $\reverse{G}$.
We infer that $\src(G \setminus (Z \setminus \{v\}), v) = \emptyset$ and $Z \setminus \{v\}$ is a multicut in $(G,\terms)$, a contradiction to the choice of $Z$.

Second, take an arbitrary vertex $u \in Y$, that is, $u$ is in the source shadow of $Z$ w.r.t. $\sources$ in $G$.
Let $B \subseteq Z$ be the set of those vertices $v\in Z$, for which there exists an $vu$-path in $G \setminus (Z \setminus \{v\})$.
Clearly, $B$ is a $u-\sources$ separator in $\reverse{G}$. We claim that it is an important one.

If $B$ is not a minimal $u-\sources$ separator in $\reverse{G}$, say $v \in B$ and $B \setminus \{v\}$ is a $u-\sources$ separator in $\reverse{G}$,
then $B \setminus \{v\}$ is a $v-\sources$ separator in $\reverse{G}$ (as there is a $vu$-path in $G \setminus (Z \setminus \{v\})$)
and $Z \setminus \{v\}$ is a multicut in $(G,\terms)$, a contradiction to the choice of $Z$.

Assume then that there exists a $u-\sources$ separator $B'$ in $\reverse{G}$ that is {\em{behind}} $B$ and $|B'| \leq |B|$, $B' \neq B$. Moreover we may assume that $B'$ is minimal;
this gives us an existence of a $Bv$-path in $\reverse{G}$ for any $v \in B'$.
We claim that $Z' = (Z \setminus B) \cup B'$ is a multicut in $(G,\terms)$. This would lead to a contradiction with the choice of $Z$,
as $|Z'| \leq |Z|$ and $Z'$ is smaller in the lexicographical order than $Z$.
Assume then that $Z'$ is not a multicut in $(G,\terms)$, that is, there is an $s_it_i$-path $P$ in $G \setminus Z'$ for some $1 \leq i \leq \nterms$.
As $Z$ is a multicut in $(G,\terms)$, $P$ contains at least one vertex $v \in B \setminus B' = Z \setminus Z'$. By the choice of $B$ and the fact
that $B'$ is behind $B$, we infer that there exists a $vu$-path in $G \setminus Z'$. As $P$ contains $v$, there exists a $s_iu$-path in $G \setminus Z'$,
a contradiction to fact that $B'$ is an $u-\sources$ separator in $\reverse{G}$. This concludes the proof of the lemma.
\end{proof}

\section{The algorithm}\label{sec:alg}

\subsection{Potential function and simple operations}\label{sec:alg:init}

Our algorithm consists of a number of branching steps. To measure the progress of the algorithm, we introduce the following potential function.
\begin{definition}[\textbf{potential}]
Given a \dagmulticut{} instance $\inst = (G,\terms,\cut)$, we define its potential $\poten(\inst)$ as
$\poten(\inst) = (\nterms+1)\cut - \sum_{i=1}^\nterms \textrm{cut}_G(s_i,t_i)$.
\end{definition}
Observe, that if $\inst=(G,\terms,\cut)$ is a \dagmulticut{} instance, in which $\textrm{cut}(s_i,t_i)>p$ for some $(s_i,t_i)\in \terms$, then we can immediately conclude that $\inst$ is a NO instance. Therefore, w.l.o.g. we can henceforth assume that $\textrm{cut}(s_i,t_i)\leq p$ for all $(s_i,t_i)\in \terms$ in all the appearing instances of \dagmulticut.

In many places we perform the following simple operations on \dagmulticut{} instances $(G,\terms,\cut)$. We formalize their properties in subsequent lemmata.
\begin{definition}[\textbf{killing a vertex}]
For a \dagmulticut{} instance $(G,\terms,\cut)$ and a nonterminal vertex $v$ of $G$,
 by {\em{killing the vertex $v$}} we mean the following operation: we delete the vertex $v$ and decrease $\cut$ by one.
\end{definition}
\begin{definition}[\textbf{bypassing a vertex}]
For a \dagmulticut{} instance $(G,\terms,\cut)$ and a nonterminal vertex $v$ of $G$,
 by {\em{bypassing the vertex $v$}} we mean the following operation: we delete the vertex $v$ and for each in-neighbour $v^-$ of $v$ and each
 out-neighbour $v^+$ of $v$ we add an arc $(v^-,v^+)$.
\end{definition}
\begin{lemma}\label{lem:kill}
Let $\inst'=(G',\terms,\cut-1)$ be obtained from \dagmulticut{} instance $\inst=(G,\terms,\cut)$ by killing a vertex $v$.
Then $\inst'$ is a YES instance if and only if $\inst$ is a YES instance that admits a solution that contains $v$.
Moreover, $\poten(\inst') < \poten(\inst)$.
\end{lemma}
\begin{proof}
Let $Z$ be a multicut in $\inst$ that contains $v$. As $G \setminus Z = G' \setminus (Z \setminus \{v\})$,
$Z \setminus \{v\}$ is a multicut in $\inst'$ of size $|Z|-1$.
In the other direction, if $Z$ is a multicut in $\inst'$, then $G' \setminus Z = G \setminus (Z \cup \{v\})$
and $Z \cup \{v\}$ is a multicut in $\inst$ of size $|Z|+1$.
To see that the potential strictly decreases, note that $\textrm{cut}_{G'}(s_i,t_i) \geq \textrm{cut}_G(s_i,t_i) - 1$ for all $1 \leq i \leq \nterms$.
\end{proof}
\begin{lemma}\label{lem:bypass}
Let $\inst'=(G',\terms,\cut)$ be obtained from \dagmulticut{} instance $\inst=(G,\terms,\cut)$ by bypassing a vertex $v$.
Then:
\begin{enumerate}
\item any multicut in $\inst'$ is a multicut in $\inst$ as well;
\item any multicut in $\inst$ that does not contain $v$ is a multicut in $\inst'$ as well;
\item $\src(G,u) = \src(G',u)$ for any $u \in V(G') = V(G) \setminus \{v\}$;
\item $\poten(\inst') \leq \poten(\inst)$.
\end{enumerate}
\end{lemma}
\begin{proof}
The lemma follows from the following observations on relations between paths in $G$ and $G'$.
For a path $P$ in $G$ whose first and last point is different than $v$, we define $P_{G'}$
as $P$ with a possible occurrence of $v$ removed. By the definition of $G'$, $P_{G'}$ is a path in $G'$.
In the other direction, for a path $P$ in $G'$, we define $P_G$ as a path obtained from $P$
by inserting the vertex $v$ between any consecutive vertices $v_i,v_{i+1}$ for which $(v_i,v_{i+1}) \in E(G') \setminus E(G)$.
Since $G$ and $G'$ are acyclic, the vertex $v$ is inserted at most once. By the construction of $G'$, $P_G$ is a path in $G$.

Now, for a multicut $Z$ in $\inst'$ and an arbitrary $s_it_i$-path $P$ in $G$, the path $P_{G'}$ is intersected
by $Z$, thus $P$ is intersected by $Z$ as well and $Z$ is a multicut in $\inst$.
For a multicut $Z$ in $\inst$ with $v \notin Z$, and an arbitrary $s_it_i$-path $P$ in $G'$, the path $P_G$
is intersected by $Z$. As $v \notin Z$, we infer that $P$ is intersected by $Z$ as well and $Z$ is a multicut in $\inst'$.

To prove the third claim, note that for any $u \in V(G')$, any $s_iu$-path $P$ in $G$ yields a $s_iu$-path $P_{G'}$ in $G'$ and vice versa.
Finally, the last claim follows from the fact that any family $\mathcal{P}$ of $s_it_i$-paths in $G$ with pairwise disjoint sets of internal vertices
can be transformed into a similar family $\mathcal{P}' = \{P_{G'}: P \in \mathcal{P}\}$ in $G'$. Thus $\mathrm{cut}_{G'}(s_i,t_i) \geq \mathrm{cut}_G(s_i,t_i)$ for $1 \leq i \leq \nterms$.
\end{proof}

We note that bypassing a vertex corresponds to the torso operation of Chitnis et al. \cite{dir-mwc} and, if we perform a series of bypass operations, the result does not depend on their order.

\begin{lemma}\label{lem:bypass-torso}
Let $\inst=(G,\terms,\cut)$ and $X \subseteq V(G)$ be a subset of nonterminals of $G$.
Let $\inst'=(G',\terms,\cut)$ be obtained from $\inst$ by bypassing all vertices of $X$ in an arbitrary order.
Then $(u,v) \in E(G')$ for $u,v \in V(G') = V(G) \setminus X$ if and only if there exists a $uv$-path in $G$ with internal vertices from $X$ (possibly consisting only of an arc $(u,v)$).
In particular, $\inst'$ does not depend on the order in which the vertices of $X$ are bypassed.
\end{lemma}

\begin{proof}
We perform induction with respect to the size of the set $X$. For $X = \emptyset$ the lemma is trivial.

Let $\inst''=(G'',\terms,\cut)$ be an instance obtained by bypassing all vertices of $X \setminus \{w\}$ in $\inst$ in an arbitrary order, for some $w \in X$.
Take $u,v \in V(G') = V(G) \setminus X$.
Assume first that $(u,v) \in E(G')$. By the definition of bypassing, $(u,v) \in E(G'')$ or $(u,w),(w,v) \in E(G'')$.
In the first case there exists a $uv$-path in $G$ with internal vertices from $X \setminus \{w\}$ by the induction hypothesis.
In the second case, by the induction hypothesis, there exist a $uw$-path and $wv$-path in $G$, both with internal vertices from $X \setminus \{w\}$; their concatenation
is the desired $uv$-path (recall that $G$ is acyclic).

In the other direction, let $P$ be a $uv$-path in $G$ with internal vertices in $X$. If $w$ does not lie on $P$, by the induction hypothesis
$(u,v) \in E(G'')$. Otherwise, $P$ splits into a $uw$-path and a $wv$-path, both with internal vertices in $X \setminus \{w\}$.
By the induction hypothesis $(u,w),(w,v) \in E(G'')$. By the definition of bypassing, $(u,v) \in E(G')$ and the lemma is proven.
\end{proof}

\subsection{Degree reduction}\label{ss:alg:deg-red}

In this section we introduce the second --- apart from the source shadow reduction in Lemma \ref{lem:shadowless} --- main tool used in the algorithm.
In an instance $(G,\terms,\cut)$, let $B_i$ be the $s_i-t_i$ mincut closest to $s_i$ and let $Z$ be a solution.
If we know that a $vt_i$-path survives in $G \setminus Z$ for some $v \in B_i$, we may add an arc $(v,t_i)$ and then bypass the vertex $v$,
strictly increasing the value $\textrm{cut}_G(s_i,t_i)$ (and thus decreasing the potential) by Lemma~\ref{lem:push-mincut}.
Therefore, we can branch: we either guess the pair $(i,v)$, or guess that none such exist;
in the latter branch we do not decrease potential
but instead we may modify the set of arcs incident to the sources to get some structure,
as formalized in the following definition.
\begin{definition}[\textbf{degree-reduced graph}]\label{def:deg-red}
For a \dagmulticut{} instance $(G,\terms,\cut)$ the {\em{degree-reduced}} graph $G^\ast$ is a graph constructed as follows.
For $1 \leq i \leq \nterms$, let $B_i$ be the $s_i-t_i$ mincut closest to $s_i$. We start with $V(G^\ast) = V(G)$, $E(G^\ast) = E(G \setminus \sources)$ and then,
for each $1 \leq i \leq \nterms$, we add an arc $(s_i,v)$ for all $v \in B_i$ and for all $v \in \bigcup_{1 \leq i' \leq \nterms} B_{i'}$ for which 
$s_i \in \src(G,v)$ but $v$ is not reachable from $B_i$ in $G$.
\end{definition}

Let us now estabilish some properties of the degree-reduced graph.

\begin{lemma}[\textbf{properties of the degree-reduced graph}]\label{lem:deg-red-prop}
For any \dagmulticut{} instance $\inst = (G,\terms,\cut)$ and 
the degree-reduced graph $G^\ast$ of $\inst$, the following holds:
\begin{enumerate}
\item $|\outN_{G^\ast}(\sources)| \leq \nterms\cut$.\label{deg-red:deg}
\item for each $1 \leq i \leq \nterms$, $B_i$ is the $s_i-t_i$ mincut closest to $s_i$ in $G^\ast$.\label{deg-red:cut-hold}
\item $\poten(\inst') = \poten(\inst)$, where $\inst' = (G^\ast,\terms,\cut)$.\label{deg-red:poten}
\item $Z \subseteq V(G)$ is a multicut in $(G^\ast,\terms)$ if and only if
$Z$ is a multicut in $(G,\terms)$ satisfying the following property:
for each $1 \leq i \leq \nterms$, for each $v \in B_i$, the vertex $v$ is either in $Z$ or $Z$ is an $v-t_i$ separator;
in particular, $\inst'$ is a YES instance if and only if $\inst$ is a YES instance that admits a solution satisfying the above property.\label{deg-red:iff}
\item for each $v \in V(G)$ we have $\src(G^\ast,v) \subseteq \src(G,v)$;
moreover, if $(s_i,v)$ is an arc in $G^\ast$ for some $1 \leq i \leq \nterms$ then $\src(G^\ast,v) = \src(G,v)$.\label{deg-red:keep-src}
\end{enumerate}
\end{lemma}
\begin{proof}
For Claim \ref{deg-red:deg}, note that for each $1 \leq i \leq \nterms$, $\outN_{G^\ast}(s_i) \subseteq \bigcup_{1 \leq i' \leq \nterms} B_{i'}$
and $|B_{i'}| = \textrm{cut}_G(s_{i'},t_{i'}) \leq \cut$.

For Claim \ref{deg-red:cut-hold}, we first prove that $B_i$ is a $s_i-t_i$ separator in $G^\ast$.
Assume otherwise; let $P$ be a $s_it_i$-path in $G^\ast$ that avoids $B_i$.
Let $(s_i,v)$ be the first arc on $P$. By the construction of $G^\ast$, in $G$ the vertex $v$ is reachable from $s_i$ but not from $B_i$.
Therefore, there exists an $s_iv$-path in $G \setminus B_i$; together with $P$ truncated by $s_i$ it gives a $s_it_i$-path in $G \setminus B_i$, a contradiction.

Now note that for each $1 \leq i \leq \nterms$, there exist $\textrm{cut}_G(s_i,t_i)$ $s_it_i$-paths in $G$ with pairwise disjoint interiors;
each path visits different vertex of $B_i$. These paths can be shortened in $G^\ast$ using arcs $(s_i,v)$ for $v \in B_i$; thus, $B_i$ is a
$s_i-t_i$ separator in $G^\ast$ of minimum size. The fact that it is an important $t_i-s_i$ separator in $\reverse{G}$ follows from
the fact that $B_i \subseteq \outN_{G^\ast}(s_i)$.

Claim \ref{deg-red:poten} follows directly from Claim \ref{deg-red:cut-hold}, as $\textrm{cut}_G(s_i,t_i) = \textrm{cut}_{G^\ast}(s_i,t_i)$ for all $1 \leq i \leq \nterms$.

For Claim \ref{deg-red:iff}, let $Z$ be a multicut in $(G^\ast,\terms)$. Take arbitrary $1 \leq i \leq \nterms$ and let $P$ be an arbitrary $s_it_i$-path in $G$.
This path intersects $B_i$; let $v$ be the last (closest to $t_i$) vertex from $B_i$ on $P$. Let $P^\ast$ be a $s_it_i$-path in $G^\ast$ defined as follows: we start
with the arc $(s_i,v)$ and then we follow $P$ from $v$ to $t_i$. As $Z$ is a multicut in $(G^\ast,\terms)$, $Z$ intersects $P^\ast$. We infer that $Z$ intersects $P$
and $Z$ is a multicut in $G$. Moreover, for each $1 \leq i \leq B_i$ and $v \in B_i$, if $v \notin Z$ then $Z$ is a $v-t_i$ separator in $G^\ast$ (and in $G$ as well, as
    $G$ and $G^\ast$ differ only on arcs incident to the sources), as otherwise $Z$ would not intersect an $s_it_i$-path in $G^\ast$ that starts with the arc $(s_i,v)$.

In the second direction, let $Z$ be a multicut in $(G,\terms)$ that satisfies the conditions given in Claim \ref{deg-red:iff}.
Let $P^\ast$ be an arbitrary $s_it_i$-path in $G^\ast$. As $B_i$ is an $s_i-t_i$ separator in $G^\ast$, $P^\ast$ intersects $B_i$; let $v$ be the last (closest to $t_i$)
vertex of $B_i$ on $P^\ast$. Note that the part of the path $P^\ast$ from $v$ to $t_i$ (denote it by $P_v$) is present also in the graph $G$. By the properties of $Z$,
$v \in Z$ or $Z$ intersects $P_v$. Thus $Z$ intersects $P^\ast$ and $Z$ is a multicut in $(G^\ast,\terms)$.

To see the first part of Claim \ref{deg-red:keep-src} note that if $(s_i,v)$ is an arc in $G^\ast$, then $s_i \in \src(G,v)$: clearly this is true for $v \in B_i$,
and otherwise $s_i \in \src(G,v)$ is one of the conditions required to add arc $(s_i,v)$.
The second part follows directly from the construction: if $(s_i,v)$ is an arc in $G^\ast$, then $v \in B_{i'}$ for some $1 \leq i' \leq \nterms$.
If $s_{i''} \in \src(G,v)$ then either $v$ is reachable from some vertex of $B_{i''}$ in $G^\ast$, or
the arc $(s_{i''},v)$ is present in $G^\ast$.
\end{proof}
We note that in the definition of the degree-reduced graph, the arcs between
a source $s_i$ and vertices in $B_{i'}$ for $i \neq i'$ are added only to ensure
Claim \ref{deg-red:keep-src}. For the remaining claims, as well as the branching 
described at the beginning of the section (formalized in the subsequent lemma)
it would suffice to add only arcs $(s_i,v)$ for $1 \leq i \leq \nterms$ and $v \in B_i$.

\begin{lemma}\label{lem:degree-branch}
There exists an algorithm that, given a \dagmulticut{} instance $\inst=(G,\terms,\cut)$, in polynomial time
generates a sequence of instances $(\inst_j=(G_j,\terms_j,\cut_j))_{j=1}^d$ satisfying the following properties. Let $\inst_0 = (G^\ast,\terms,\cut)$;
\begin{enumerate}
\item if $Z$ is a multicut $Z$ in $\inst_j$ for some $0 \leq j \leq d$, then $Z \subseteq V(G)$ and $Z$ is a multicut in $\inst$ too;
\item for any multicut $Z$ in $\inst$, there exists $0 \leq j \leq d$ such that $Z$ is a multicut in $\inst_j$ too;
\item for each $1 \leq j \leq d$, $\cut_j = \cut$, $\terms_j = \terms$ and $\poten(\inst_j) < \poten(\inst)$;
\item $d \leq \nterms\cut$.
\end{enumerate}
\end{lemma}
\begin{proof}
Let $B_i$ be as in Definition \ref{def:deg-red}.
Informally speaking, we guess an index $1 \leq i \leq \nterms$ and a vertex $v \in B_i$
such that $t_i$ is reachable from $v$ in $G \setminus Z$, where $Z$ is a solution to $\inst$ (in particular, $v \notin Z$). If we have such $v$,
we can add an arc $(v,t_i)$ and then bypass $t_i$; by the choice of $B_i$ and Lemma \ref{lem:push-mincut}, the value $\textrm{cut}(s_i,t_i)$ strictly increases during this operation.
The last branch --- where such a choice $(i,v)$ does not exists --- corresponds to the degree-reduced graph $G^\ast$. We now proceed to the formal arguments.

For $1 \leq i \leq \nterms$ and $v \in B_i$ we define the graph $G_{i,v}$ as follows: we first add an arc $(v,t_i)$ to $G$ and then bypass the vertex $v$.
We now apply Lemma \ref{lem:push-mincut} for $t_i-s_i$ cuts in the graph $\reverse{G}$: $B_i$ is the unique minimum size important $t_i-s_i$ separator,
$v \in B_i$ and $\reverse{G_{i,v}}$ contains the same set of vertices and a superset of arcs of the graph $G'$ from the statement of the lemma.
Therefore $\textrm{cut}_{G_{i,v}}(s_i,t_i) > \textrm{cut}_G(s_i,t_i)$.
Moreover, $\textrm{cut}_{G_{i,v}}(s_{i'},t_{i'}) \geq \textrm{cut}_G(s_{i'},t_{i'})$ for $i' \neq i$, as adding an arc and bypassing a vertex
cannot decrease the size of the minimum separator. Therefore $\poten((G_{i,v},\terms,\cut)) < \poten(\inst)$;
we output $(G_{i,v},\terms,\cut)$ as one of the output instances $\inst_j$. Clearly, $d = \sum_{i=1}^\nterms |B_i| \leq \nterms\cut$.
To finish the proof of the lemma we need to show the equivalence stated in the first two points of the statement.

In one direction, note that as the graphs $G_{i,v}$ are constructed from $G$ by adding an arc and bypassing a vertex, any multicut in $(G_{i,v},\terms)$ is a multicut in $(G,\terms)$
as well. Moreover, by Lemma \ref{lem:deg-red-prop}, Claim \ref{deg-red:iff}, any multicut of $\inst_0$ is a multicut of $\inst$ as well.

In the other direction, let $Z$ be a solution to $\inst$. Consider two cases. Firstly assume that there exists $1 \leq i \leq \nterms$ and $v \in B_i$
such that $v \notin Z$ and $Z$ is {\em{not}} a $v-t_i$ separator. As $Z$ is a multicut in $(G,\terms)$, $Z$ is a $s_i-v$ separator in $G$.
Therefore $Z$ is also a multicut in a graph $G$ with the arc $(v,t_i)$ added. As $v \notin Z$, by Lemma \ref{lem:bypass},
$Z$ is a multicut in $(G_{i,v},\terms)$.
In the second case, if for each $1 \leq i \leq \nterms$ and $v \in B_i$, we have $v \in Z$ or $Z$ is a $v-t_i$ separator in $G$,
   we conclude that $Z$ is a multicut in $(G^\ast,\terms)$ by Lemma \ref{lem:deg-red-prop}, Claim \ref{deg-red:iff}.
\end{proof}

\subsection{Overview on the branching step}

In order to prove Theorem \ref{thm:main}, we show the following lemma that encapsulates a single branching step of the algorithm.
\begin{lemma}\label{lem:branching-step}
There exists an algorithm that, given a \dagmulticut{} instance $\inst = (G,\terms,\cut)$ with $|\terms| = \nterms$, in time $\Ohstar(2^{\nterms + 2^{\Oh(\cut)}})$
either correctly concludes that $\inst$ is a NO instance,
or computes a sequence of instances $(\inst_j = (G_j,\terms_j,\cut_j))_{j=1}^d$ such that:
\begin{enumerate}
\item $\inst$ is a YES instance if and only if at least one instance $\inst_j$ is a YES instance;
\item for each $1 \leq j \leq d$, $V(G_j) \subseteq V(G)$, $\cut_j \leq \cut$, $\terms_j = \terms$ and $\poten(\inst_j) < \poten(\inst)$;
\item $d \leq 4 \cdot 2^{\nterms + 2^{\Oh(\cut)}} \nterms\cut \log |V(G)|$.
\end{enumerate}
\end{lemma}

The algorithm of Theorem \ref{thm:main} applies Lemma \ref{lem:branching-step}
and solves the output instances recursively.

\begin{proof}[Proof of Theorem \ref{thm:main}]
Let $\inst = (G,\terms,\cut)$ be a \dagmulticut{} instance. Clearly, if $\poten(\inst) < 0$ then $\textrm{cut}(s_i,t_i) > \cut$ for some $1 \leq i \leq \nterms$
and the instance is a NO instance. Otherwise, we apply Lemma \ref{lem:branching-step}, and solve the output instances recursively.
Note that the potential of $\inst$ is an integer bounded by $(\nterms+1)\cut$, thus the search tree of the algorithm has depth at most $(\nterms+1)\cut$.
Using a simple fact that for $k,n > 1$ we have
$$\log^k n = 2^{k \log\log n} \leq 2^{\frac{2}{3}k^{3/2} + \frac{1}{3}(\log\log n)^3} = 2^{\frac{2}{3}k^{3/2}} n^{o(1)},$$
we obtain that the number of leaves of the search tree is bounded by
\begin{eqnarray*}
\left( 2^{\nterms + 2^{\Oh(\cut)}} \nterms\cut \log|V(G)| \right)^{(\nterms+1)\cut} & = & 2^{\Oh(\nterms^2\cut)}\cdot 2^{\Oh(\nterms 2^{\Oh(\cut)})} \cdot 2^{\Oh(\nterms^{3/2}\cut^{3/2})} |V(G)|^{o(1)} \\
& = & \Ohstar(2^{\Oh(\nterms^2\cut+\nterms 2^{\Oh(\cut)})}).
\end{eqnarray*}
The last equality follows from the fact that $\nterms^{3/2}\cut^{3/2}\leq \nterms^2\cut+\nterms\cut^2\leq \nterms^2\cut+\nterms2^{\cut}$.
\end{proof}

In rough overview of the proof of Lemma \ref{lem:branching-step}, we describe a sequence of steps
where in each step, either the potential of the instance is decreased
or more structure is forced onto the instance. For example, consider
Lemma \ref{lem:degree-branch}. The result is a branching into polynomially many branches,
where in every branch but one the potential strictly decreases, and in
the remaining branch, the degrees of the source terminals are bounded.
Thus we may treat this step as ``creating'' a degree-reduced instance.

In somewhat more detail, let $Z$ be the lex-min solution to $\inst$.
We guess a set~$S \subseteq T^s$
such that there is some~$v \in Z$ with~$S(G,v)=S$, but no~$v'\in Z$
with~$S(G,v')\varsubsetneq S$;
bypass any vertex~$u$ with~$S(G,u) \varsubsetneq S$. 
By appropriately combining degree reduction with shadow removal, we
may further assume that no vertex in~$V(G,S)$ is in source-shadow
of~$Z$, and that the sources~$S$ have bounded degree.
Consider now the first vertex~$v \in V(G,S)$ under~$\toporder$ (if
any) which has its set of seen sources modified by~$Z$, i.e., 
$v \in V(G,S) \setminus Z$,~$S(G \setminus Z, v) \varsubsetneq
S$, and~$v$ is~$\toporder$-minimal among all such vertices. 
Let~$w$ be an in-neighbour of~$v$. The important observation is that
since~$S(G,w)$ is by assumption not modified by~$Z$, 
every such vertex~$w$ must be either a source or
deleted. Since~$v$ is not in source shadow of $Z$, 
there is an arc~$(s,v)$ in~$G$ for some~$s \in S$, and by the degree
reduction, there is only a bounded number of such vertices. 
Thus, if any vertex is modified by~$Z$ in this sense, then we may 
find one by branching on the out-neighbours of~$S$, decreasing the potential.

Otherwise, we know that if~$v \in V(G,S)$, then either~$v
\in Z$ or~$S(G \setminus Z, v)=S$. Thus, we may ``flatten'' the graph,
by making every~$v \in V(G,S)$ a direct out-neighbour of every~$s\in S$. 
By further degree reduction, we can now identify a polynomially sized
set out of which at least one vertex must be deleted. 

\subsection{Branchings and reductions}\label{ss:branching-step}

We now proceed with the formal proof of Lemma \ref{lem:branching-step}.
The proof contains a sequence of {\em{branching rules}} (when we generate a number of subcases,
some of them already ready to output as one instance $\inst_j$), or {\em{reduction rules}} (when we reduce the graph without changing the answer).
To make the algorithm easier to follow, we embed all branching and reduction rules in appriopriately numbered environments.

If the input instance $\inst$ is YES instance, by $Z$ we denote its lex-min solution.
Whenever we perform a branching or reduction step, in the new instance we consider the topological order that is induced by the old one; all the reductions and branchings add arcs only directed from vertices smaller in $\toporder$ to bigger. This also ensures that during the course of the algorithm all the directed graphs in the instances are acyclic.

We start with the obvious rule that was already mentioned in Section~\ref{sec:alg}.
Then, we roughly localize one vertex of $Z$.
\begin{reduction}\label{red:large-cut}
If $\textrm{cut}_G(s_i,t_i) > \cut$ (in particular, if $(s_i,t_i) \in E(G)$)
  for some $1 \leq i \leq \nterms$, conclude that $\inst$ is a NO instance.
\end{reduction}
\begin{branching}\label{branch:S}
Branch into $2^\nterms-1$ subcases, labeled by nonempty sets $S \subseteq \sources$.
In the case labeled $S$ we assume that $Z$ contains a vertex $v$ with $\src(G,v) = S$, but
no vertex $v'$ with $\src(G,v')$ being a proper subset of $S$.
\end{branching}
As $Z$ is a lex-min solution (in case of $\inst$ being a YES instance), $Z$ cannot contain any vertex $v$ with $\src(G,v) = \emptyset$. In each branch we can bypass some vertices.
\begin{reduction}\label{red:bypass-subset-S}
In each subcase, as long as there exists a nonterminal vertex $u \in V(G)$ with $\src(G,u) \varsubsetneq S$ bypass $u$. Let $(G^1,\terms,\cut)$ be the reduced instance.
\end{reduction}
By Lemma \ref{lem:bypass}, an application of the above rule cannot turn a NO instance into a YES instance. Moreover, in the branch where $S$ is guessed correctly,
$Z$ remains the lex-min solution to $(G^1,\terms,\cut)$. By Lemma \ref{lem:bypass},
$\poten((G^1,\terms,\cut)) \leq \poten(\inst)$.

We now apply the reduction of source degrees.
\begin{branching}\label{branch:deg-red-1}
In each subcase, let $S$ be its label and $(G^1,\terms,\cut)$ be the instance.
Invoke Lemma \ref{lem:degree-branch} on the instance $(G^1,\terms,\cut)$.
Output all instances $\inst_j$ for $1 \leq j \leq d$ as part of the output instances in Lemma \ref{lem:branching-step}.
Keep the instance $\inst_0$ for further analysis in this subcase and denote $\inst_0 = (G^2,\terms,\cut)$;
$G^2$ is the degree-reduced graph of $G^1$.
\end{branching}
Let us summarize what Lemma \ref{lem:degree-branch} implies on the outcome of Branching \ref{branch:deg-red-1}.
We output at most $2^\nterms \nterms \cut$ instances, and keep one instance
for further analysis in each branch.
Each output instance has strictly decreased potential, while $\poten((G^2,\terms,\cut)) \leq \poten(G^1,\terms,\cut))$. If $\inst$ is a NO instance,
all the generated instances --- both the output and kept ones --- are NO instances.
If $\inst$ is a YES instance,
   then it is possible that all the output instances are NO instances only if
in the branch where the set $S$ is guessed correctly,
the solution $Z$ is a solution to $(G^2,\terms,\cut)$ as well.
Moreover, as any solution to $(G^2,\terms,\cut)$ is a solution to $\inst$
as well by Lemma \ref{lem:deg-red-prop},
in this case $Z$ is the lex-min solution to $(G^2,\terms,\cut)$.

Let us now investigate more deeply the structure of the kept instances.
\begin{lemma}\label{lem:src-anal}
In a branch, let $S$ be its label, $(G^1,\terms,\cut)$ the instance on which Lemma \ref{lem:degree-branch} is invoked
and $(G^2,\terms,\cut)$ the kept instance.
For any $v \in V(G^1) = V(G^2)$ with $\src(G,v) = S$, we have $\src(G^1,v) = S$ and
$\src(G^2,v) \in \{\emptyset, S\}$.
\end{lemma}
\begin{proof}
Note that the operation of bypassing a vertex $u$ does not change whether a vertex $v$ is reachable from a fixed source;
thus $\src(G,v) = \src(G^1,v) = S$. By Lemma \ref{lem:deg-red-prop}, we have $\src(G^2,v) \subseteq \src(G^1,v) = S$.
Assume that $\src(G^2,v) \neq \emptyset$, let $s_i \in \src(G^2,v)$, let $P$ be a $s_iv$-path in $G^2$
and let $(s_i,w)$ be the first arc on this path.
Since $G^2$ differs from $G^1$ only on the set of arcs incident to the set of sources, the subpath $P'$
of $P$ from $w$ to $v$ is present in $G^1$ as well. Therefore $\src(G^1,w) \subseteq \src(G^1,v) = S$.
As $w$ was not bypassed by Reduction \ref{red:bypass-subset-S}, we have $\src(G,w) = S$.
Using again the fact that bypassing a vertex $u$ does not change whether $w$ is reachable from a fixed source, we have that $\src(G^1,w) = S$.
By Lemma \ref{lem:deg-red-prop}, $\src(G^1,w) = \src(G^2,w) = S$.
By the presence of $P'$ in $G^2$, we have $S \subseteq \src(G^2,v)$. This finishes the proof of the lemma.
\end{proof}

Recall that if $\inst$ is a YES instance, but all instances output so far are NO instances,
then in some subcase $S$ the set $Z$ is the lex-min solution to $(G^2,\terms,\cut)$.
In this case $Z$ does not contain any vertex from $V(G^2,\emptyset)$
and we can remove these vertices,
as they are not contained in any $s_it_i$-path for any $1 \leq i \leq \nterms$.
\begin{reduction}\label{red:delete-futile-1}
In each branch, let $S$ be its label and $(G^2,\terms,\cut)$ the kept instance.
As long as there exists a nonterminal vertex $v \in V(G^2)$ with $\src(G^2,v) = \emptyset$, delete $v$.
Denote the output instance by $(G^3,\terms,\cut)$.
\end{reduction}

Reduction \ref{red:delete-futile-1} does not interfere with any $s_it_i$-paths, thus $\poten((G^3,\terms,\cut)) = \poten((G^2,\terms,\cut))$.
Again, if $\inst$ is a NO instance, all instances $(G^3,\terms,\cut)$ are NO instances as well,
  and if $\inst$ is a YES instance, but all output instances produced so far are NO instances, $Z$
  is the lex-min solution to $(G^3,\terms,\cut)$ in some branch $S$.
Moreover, in $G^3$ each source has out-degree at most $\nterms\cut$ and there is no vertex $v$ with $\src(G^3,v) \varsubsetneq S$
(note that Reduction \ref{red:delete-futile-1} does not change reachability of a vertex from a fixed source).
We apply the source shadow reduction to $(G^3,\terms,\cut)$.
\begin{branching}\label{branch:shadowless}
In each branch, let $S$ be its label, and $(G^3,\terms,\cut)$ be the remaining instance.
Invoke Lemma \ref{lem:shadowless} on $(G^3,\terms,\cut)$, obtaining a family $\shfam_S$.
Branch into $|\shfam_S|$ subcases, labeled by pairs $(S,\shadow)$ for $\shadow \in \shfam_S$.
In each subcase, obtain a graph $(G^4,\terms,\cut)$ by bypassing (in arbitrary order)
all vertices of $\shadow \setminus \outN_{G^3}(\sources)$.
\end{branching}
Note that the graph $G^4$ does not depend on the order in which we bypass vertices of $\shadow \setminus \outN_{G^3}(\sources)$.
By Lemma \ref{lem:bypass}, bypassing some vertices cannot turn a NO instance into a YES instance.
Moreover, by Lemma \ref{lem:shadowless}, if $(G^3,\terms,\cut)$ is a YES instance and $Z$ is the
lex-min solution to $(G^3,\terms,\cut)$,
then there exists $\shadow \in \shfam_S$ that contains
all vertices of source shadows of $Z$, but no vertex of $Z$.
Note that no out-neighbour of a source may be contained in a source shadow;
therefore, $\shadow \setminus \outN_{G^3}(\sources)$ contains all vertices
of source shadows of $Z$ as well.
We infer that in the branch $(S,\shadow)$, $(G^4,\terms,\cut)$ is a YES instance and,
  as bypassing a vertex only shrinks the set of solutions, $Z$ is still the lex-min
  solution to $(G^4,\terms,\cut)$.
  Moreover, there are no source shadows of $Z$ in $(G^4,\terms,\cut)$.

At this point we have at most $2^{\nterms + 2^{\Oh(\cut)}} \log |V(G)|$ subcases and at most $2^\nterms \nterms\cut$ already output
instances. In each subcase, we have $\poten((G^4,\terms,\cut)) \leq \poten((G^3,\terms,\cut))$ by Lemma \ref{lem:bypass}. The following observation is crucial for further branching.
\begin{lemma}\label{lem:magic}
Take an instance $(G^4,\terms,\cut)$ obtained in a branch labeled with $(S,\shadow)$. Assume that $(G^4,\terms,\cut)$ is a YES instance
and let $Z$ be its lex-min solution. Moreover, assume that there are no source shadows
of $Z$ in $(G^4,\terms,\cut)$.
Then the following holds: if there exists a vertex $v' \in (V(G,S) \cap V(G^4))\setminus Z$ with $\src(G^4 \setminus Z, v') \neq S$,
then the first such vertex in the topological order $\toporder$ (denoted $v$) belongs to $\outN_{G^4}(\sources)$. Moreover,
     $v$ has at least one in-neighbour in $G^4$ that is not in $\sources$, and all such in-neighbours belong to $Z$.
\end{lemma}
\begin{proof}
Let $v$ be as in the statement of the lemma.
Assume there exists an in-neighbour $w$ of $v$ that is not in $\sources$ nor in $Z$.
From the previous steps of the algorithm we infer that $\src(G^4,v) = \src(G^4,w) = S$.
Moreover, $w \in V(G,S)$ as $v \in V(G,S)$: the vertex $v$ is reachable from $w$ in $G$ (possibly via vertices bypassed in Branching \ref{branch:shadowless}), but all
vertices $u$ with $S(G,u) \varsubsetneq S$ are bypassed in Reduction~\ref{red:bypass-subset-S}.
Since $w$ is earlier in $\toporder$ than $v$, from the minimality of $v$, we have $\src(G^4 \setminus Z, w) = S$, a~contradiction.

As there are no source shadows of $Z$ in $(G^4,\terms,\cut)$, there exists $s_i \in S$ and a $s_iv$-path in $G^4$
that avoids $Z$. As $v$ has no in-neighbours outside $\sources$ and $Z$, this path consists of a single arc $(s_i,v)$ and $v \in \outN_{G^4}(\sources)$.
Moreover, if all in-neighbours of $v$ in $G^4$ are sources, $\src(G^4\setminus Z,v) = \src(G^4,v) = S$, a contradiction.
\end{proof}

\begin{branching}\label{branch:magic}
In each branch, let $(S,\shadow)$ be its label and $(G^4,\terms,\cut)$ the remaining instance.
Output at most $\nterms\cut$ instances $\inst_v$, labeled by vertices $v \in \outN_{G^4}(\sources) \cap V(G,S)$ for which $\inN_{G^4}(v) \not\subseteq \sources$:
the instance $\inst_v$ is created from $(G^4,\terms,\cut)$ by killing all non-terminal in-neighbours of $v$ and bypassing $v$.
Moreover, create one remaining instance $(G^5,\terms,\cut)$ as follows: delete from $G^4$ all arcs that have their ending vertices
in $V(G,S) \cap V(G^4)$ and for each $v \in V(G,S) \cap V(G^4)$ and $s_i \in S$ add an arc $(s_i,v)$.
\end{branching}
By Lemmata \ref{lem:kill} and \ref{lem:bypass},
   the output instances have strictly smaller potential than $(G^4,\terms,\cut)$
and are NO instances if $(G^4,\terms,\cut)$ is a NO instance.
On the other hand, assume that $(G^4,\terms,\cut)$ is a YES instance with lex-min solution $Z$
such that there are no source shadows of $Z$.
If there exist vertices $v'$ and $v$ as in the statement of Lemma \ref{lem:magic},
then the instance $\inst_v$ is computed and $Z \setminus \inN_{G^4}(v)$ (i.e., $Z$ without the killed vertices) is a solution
to $\inst_v$. Otherwise, we claim that $(G^5,\terms,\cut)$ represents the remaining case:
\begin{lemma}\label{lem:G4G5}
Let $(G^4,\terms,\cut)$ be an instance obtained in the branch $(S,\shadow)$.
\begin{enumerate}
\item $\poten((G^5,\terms,\cut)) \leq \poten((G^4,\terms,\cut))$.
\item Any multicut $Z$ in $(G^5,\terms,\cut)$ is a multicut in $(G^4,\terms,\cut)$ as well.
\item Assume additionally that $(G^4,\terms,\cut)$ is a YES instance whose lex-min solution $Z$
satisfies the following properties:
there are no source shadows of $Z$ and for each $v \in V(G,S) \cap V(G^4)$, either $v \in Z$ or $\src(G^4 \setminus Z,v) = S$.
Then $(G^5,\terms,\cut)$ is a YES instance and $Z$ is its lex-min solution.
\end{enumerate}
\end{lemma}
\begin{proof}
For the first and second claim, consider an arbitrary $s_it_i$-path $P$ in $(G^4,\terms,\cut)$ for some $1 \leq i \leq \nterms$.
If $P$ does not contain any vertex from $V(G,S) \cap V(G^4)$, $P$ is present in $G^5$ as well and $Z$ intersects $P$.
Otherwise, as $V(G,S) \cap V(G^4) \subseteq V(G^4,S)$, we have that $s_i \in S$. Let $v$ be the last (closest to $t_i$) vertex on $P$ that belongs to $V(G,S) \cap V(G^4)$.
Note that $(s_i,v)$ is an arc of $G^5$; therefore a path $P'$ that starts with $(s_i,v)$ and then follows $P$ to $t_i$ is present in $G^5$.
To prove the second point, note that as $Z$ is a multicut in $(G^5,\terms,\cut)$, $Z$ intersects $P'$ and we infer that $Z$ intersects $P$.
To prove the first point, note that the above reasoning show that any set $\mathcal{P}$ of $s_it_i$-paths in $G^4$ with pairwise disjoint interiors
yields a family of the same number of $s_it_i$-paths in $G^5$, again with pairwise disjoint interiors. Therefore, for any $1 \leq i \leq \nterms$
we have $\textrm{cut}_{G^4}(s_i,t_i) \leq \textrm{cut}_{G^5}(s_i,t_i)$.

For the third claim, it is sufficient to prove that in the considered case the set $Z$ is a multicut in $(G^5,\terms,\cut)$; its minimality follows from the second point.
Consider an arbitrary $s_it_i$-path $P'$ in $(G^5,\terms,\cut)$ for some $1 \leq i \leq \nterms$.
Again, if $P'$ does not contain any vertex from $V(G,S) \cap V(G^4)$, $P'$ is present in $G^4$ as well and $Z$ intersects $P'$.
Otherwise, $s_i \in S$ and $P'$ starts with an arc $(s_i,v)$ for some $v \in V(G,S) \cap V(G^4)$.
Note that for any $v' \in V(G,S) \cap V(G^4)$, by construction we have $\inN_{G^5}(v') = S$.
Therefore $v$ is the only vertex of $P'$ that belongs to $V(G,S) \cap V(G^4)$.

If $v \in Z$, $P'$ is intersected by $Z$ in $G^5$ and we are done. Otherwise, by the assumptions on $Z$, there exists a $s_iv$-path $P$ in $G^4 \setminus Z$.
A concatenation of $P$ and $P'$ without the arc $(s_i,v)$ yields a $s_it_i$-path in $G^4$. As $Z$ is a multicut in $(G^4,\terms,\cut)$, we infer that $P'$ is intersected
by $Z$ outside $V(G,S) \cap V(G^4)$ and the lemma is proven.
\end{proof}
The structure of $V(G,S) \cap V(G^5)$ is quite simple in $(G^5,\terms,\cut)$.
Recall that, if $\inst$ is a YES instance, but no instance output so far is a YES instance, then in at least one branch $(S,\shadow)$
we have that the lex-min solution $Z$ to $\inst$ is the lex-min solution to $(G^5,\terms,\cut)$
and $Z \cap V(G,S) \cap V(G^5) \neq \emptyset$.
We would like to guess one vertex of $Z \cap V(G,S) \cap V(G^5)$. Although, $V(G,S) \cap V(G^5)$ may still be large, each vertex $v \in V(G,S) \cap V(G^5)$ has
$\inN_{G^5}(v) = S$. Therefore we may limit the size of $V(G,S) \cap V(G^5)$ by applying once again the degree reduction branching.
\begin{branching}\label{branch:deg-red-2}
In each branch, let $(S,\shadow)$ be its label and $(G^5,\terms,\cut)$ the remaining instance.
Apply Lemma \ref{lem:degree-branch} on $(G^5,\terms,\cut)$, obtaining a sequence of instances $(\inst_j)_{j=1}^d$
and the remaining instance $(G^6,\terms,\cut)$, where $G^6$ is the degree-reduced graph $G^5$.
Output all instances $\inst_j$ for $1 \leq j \leq d$ and keep $(G^6,\terms,\cut)$ for further analysis.
\end{branching}
By Lemma \ref{lem:degree-branch}, if $(G^5,\terms,\cut)$ is a NO instance, all the output instances as well as $(G^6,\terms,\cut)$
are NO instances. Otherwise, if $(G^5,\terms,\cut)$ is a YES instance with the lex-min solution $Z$,
but the instances $\inst_j$ are all NO instances, then $Z$ is the lex-min solution to $(G^6,\terms,\cut)$.

Note that, by Lemma \ref{lem:degree-branch}, all output instances have potential strictly smaller than $\poten((G^5,\terms,\cut))$, whereas
$\poten((G^6,\terms,\cut)) = \poten((G^5,\terms,\cut))$.
Moreover, applications of Branching \ref{branch:deg-red-2} in all subcases output in total at most $2^{ \nterms + 2^{\Oh(\cut)}} \nterms\cut \log|V(G)|$ instances.

We are left with the final observation.
\begin{lemma}\label{lem:small-final}
In each subcase, let $(S,\shadow)$ be its label and $(G^6,\terms,\cut)$ the remaining instance.
Then at most $\nterms\cut$ vertices $v \in V(G,S) \cap V(G^6)$ have $\src(G^6,v) \neq \emptyset$.
\end{lemma}
\begin{proof}
Note that $V(G^4) = V(G^5) = V(G^6)$. Take $v \in V(G,S) \cap V(G^6)$.
Recall that $\inN_{G^5}(v) = S$ and $G^6$ differs from $G^5$ only on the set of arcs incident to the sources,
so $\src(G^6,v) = \inN_{G^6}(v)$. The lemma follows from Lemma \ref{lem:deg-red-prop}, Claim \ref{deg-red:deg}.
\end{proof}
We may now perform once again Reduction \ref{red:delete-futile-1}:
\begin{reduction}\label{red:delete-futile-2}
In each branch, let $(S,\shadow)$ be its label and $(G^6,\terms,\cut)$ be the remaining instance.
As long as there exists a nonterminal vertex $v \in V(G^6)$ with $\src(G^6,v) = \emptyset$, delete $v$.
Denote the output instance by $(G^7,\terms,\cut)$.
\end{reduction}
As in the case of Reduction \ref{red:delete-futile-1}, $Z$
is the lex-min solution to $(G^6,\terms,\cut)$ if and only if $Z$
is the lex-min solution to $(G^7,\terms,\cut)$. Moreover,
$\poten((G^6,\terms,\cut)) = \poten((G^7,\terms,\cut))$.

By Lemma \ref{lem:small-final}, $|V(G,S) \cap V(G^7)| \leq \nterms\cut$. Now we can perform final branching.
\begin{branching}\label{branch:final}
In each subcase, let $(S,\shadow)$ be its label and $(G^7,\terms,\cut)$ the remaining instance.
For each $v \in V(G,S) \cap V(G^7)$ output an instance $\inst_v$ created from $(G^7,\terms,\cut)$ by
killing the vertex $v$.
\end{branching}
Note that if $V(G,S)\cap V(G^7)=\emptyset$, then this rule results in no branches created.

By Lemma \ref{lem:kill}, if $(G^7,\terms,\cut)$ is a NO instance, so are the output instances $\inst_v$.
On the other hand, assume that $\inst$ is a YES instance with the lex-min solution $Z$.
Then in at least one subcase $(S,\shadow)$, if no previously output instance is a YES instance,
then the instance $(G^7,\terms,\cut)$ is a YES instance, $Z$ is its lex-min solution,
and $Z \cap V(G,S) \cap V(G^7) \neq \emptyset$. Then the instance $\inst_v$ for any $v \in Z \cap V(G,S) \cap V(G^7)$
is a YES instance; in particular, $V(G,S)\cap V(G^7)$ is nonempty.
To conclude the proof of Lemma \ref{lem:branching-step} note that $\poten(\inst_v) < \poten((G^7,\terms,\cut))$ for each output instance $\inst_v$.

We conclude with a short summary and a diagram (Figure \ref{fig:diagram}) of the branching step.

\begin{figure}[htp]
\centering
\includegraphics[scale=1]{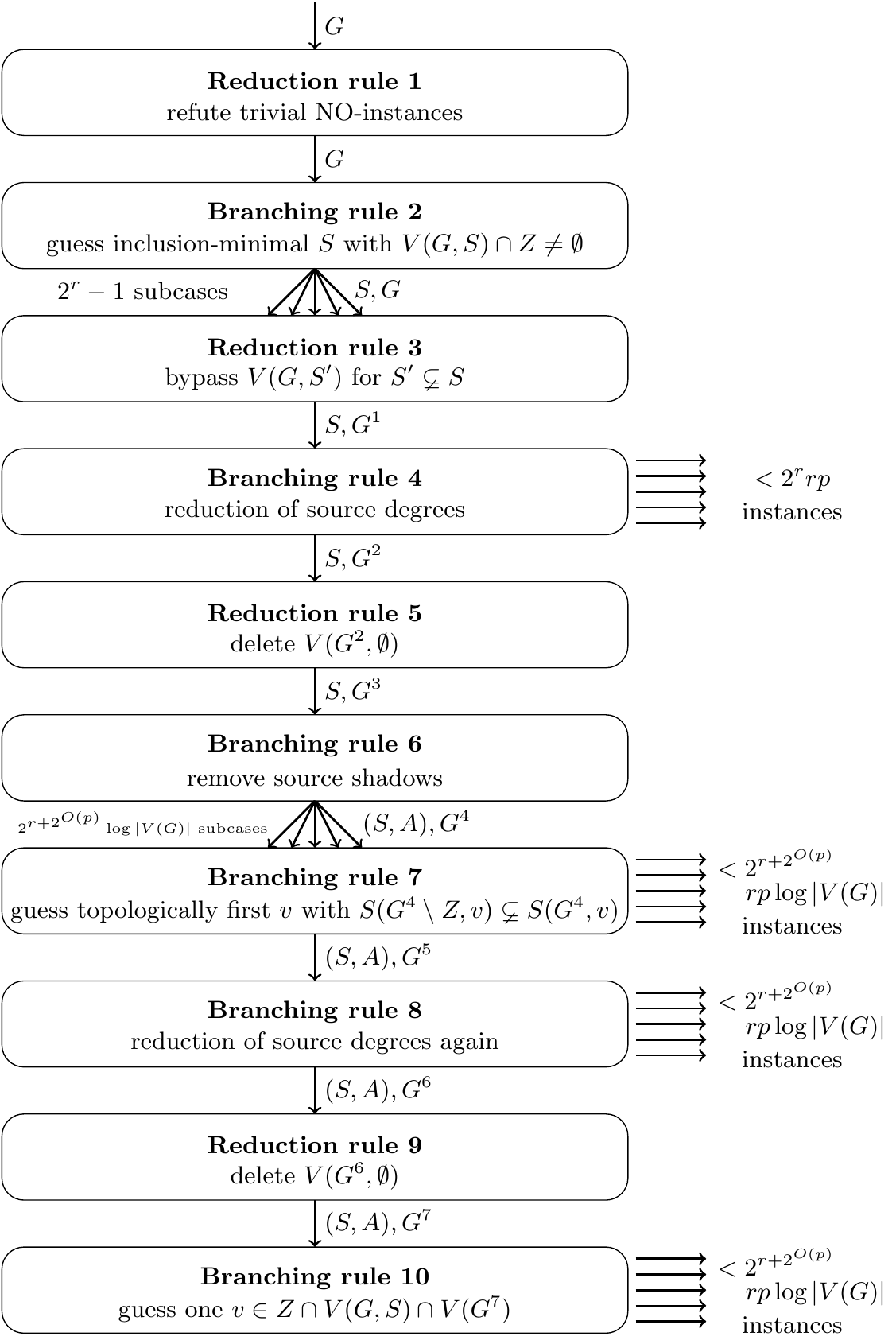}
\caption{A summary of one branching step.}
\label{fig:diagram}
\end{figure}

\begin{enumerate}
\item Branching \ref{branch:S} results in $2^\nterms-1$ subcases.
\item Branching \ref{branch:deg-red-1} outputs at most $\nterms\cut$ instances and leaves one remaining instance in each subcase; less than $2^\nterms\nterms\cut$ output instances in total.
\item Branching \ref{branch:shadowless} results in $2^{2^{\Oh(\cut)}} \log |V(G)|$ further subcases in each subcase; we have less than
$2^{\nterms+2^{\Oh(\cut)}} \log |V(G)|$ subcases at this point.
\item Branching \ref{branch:magic} outputs at most $\nterms\cut$ instances in each subcase and leaves one remaining instance;
less than $2^{\nterms+2^{\Oh(\cut)}} \nterms\cut \log |V(G)|$ output instances in total.
\item Branching \ref{branch:deg-red-2} outputs at most $\nterms\cut$ instances in each subcase and leaves one remaining instance; 
less than $2^{\nterms+2^{\Oh(\cut)}} \nterms\cut \log |V(G)|$ output instances in total.
\item Branching \ref{branch:final} outputs at most $\nterms\cut$ instances in each subcase and leaves no remaining instances; 
less than $2^{\nterms+2^{\Oh(\cut)}} \nterms\cut \log |V(G)|$ output instances in total.
\end{enumerate}

%% file: lower-bounds.tex
\section{Lower bounds}\label{sec:lb}

\subsection{W[1]-hardness of \dagmulticut{} parameterized by the size of the cutset}

The proof of Theorem \ref{thm:lb-w1} closely follows the lines of the proof of W[1]-hardness
of the general case of \multicut{} in directed graphs of Marx and Razgon \cite{marx:multicut}.
We simply need to replace the gadget $G_{i,j}$ (which is basically a long cycle)
with its acyclic variant (depicted on Figure \ref{fig:lb-w1}).
For the sake of completeness, we include here a full proof.

\begin{figure}[ht]
\centering
\includegraphics[scale=1]{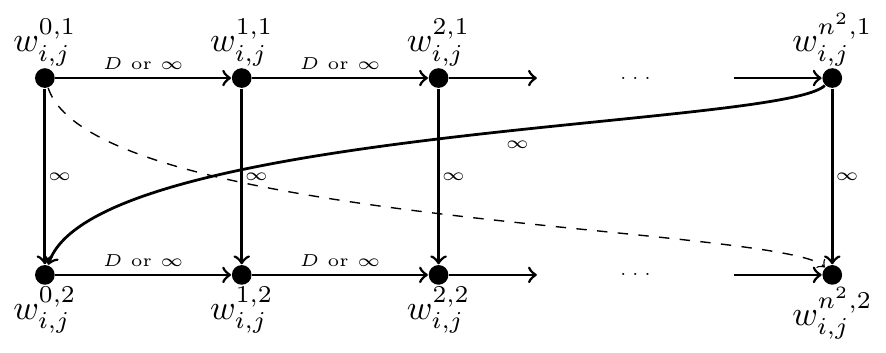}
\caption{An acyclic variant of the gadget $G_{i,j}$. Dashed arc represents the terminal pair.}
\label{fig:lb-w1}
\end{figure}

\begin{proof}[Proof of Theorem \ref{thm:lb-w1}]
We show a polynomial-time reduction from the \textsc{Clique} problem, known to be $W[1]$-hard.
Let $(G,t)$ be a \textsc{Clique} instance (i.e., we ask for a clique of size $t$
in the graph $G$). Denote $|V(G)| = n$ and $|E(G)| = m$.

Similarly as in \cite{marx:multicut}, we prove W[1]-hardness
of weighted edge-deletion variant of \dagmulticut{}.
The edge- and node-deletion variants are easily seen to be equivalent
(cf. \cite{dir-mwc}).
In our construction we use three weights: light (one),
heavy (polynomial in $t$) and infinite (that could be implemented as
budget for cuts, $p$, plus one; $p$ will be polynomial in $t$).
Therefore, all weights are polynomial in $t$,
and the weighted variant can be easily reduced to the unweighted one
by replacing arc $(u,v)$ of weight $\omega$ with
$\omega$ $uv$-paths of length two.

For each ordered pair $(i,j)$, $1 \leq i,j \leq t$, $i \neq j$,
we construct a gadget $G_{i,j}$ that has $2m$ states that encode a choice of one
pair of adjacent vertices $(v_i,v_j)$ of the desired clique in $G$.
We would like to ensure that the gadgets $G_{i,j}$ encode a clique $\{v_1,v_2,\ldots,v_t\}$
in $G$.
As discussed in \cite{marx:multicut}, it suffice to connect the gadgets
in the way to ensure that
\begin{enumerate}
\item if $G_{i,j}$ represents $(v_i,v_j)$ then $G_{j,i}$ represents $(v_j,v_i)$;
\item if $G_{i,j}$ represents $(v_i,v_j)$ and $G_{i,j'}$ represents
$(u_i,u_j)$ then $v_i=u_i$.
\end{enumerate}
In particular, it follows from the above that 
if $G_{i,j}$ represents $(v_i,v_j)$ and $G_{i',j}$ represents
$(u_i,u_j)$ then $v_j=u_j$.

Let $D = 2(t+1)^2$ be the weight of a heavy arc.
We set the budget for cuts as $p := 2t(t-1)D + t(t+1)/2$. Note that
$p < 2t(t-1)D + D$; thus we are allowed to cut only $2t(t-1)$ heavy arcs.

We now describe the gadget $G_{i,j}$, depicted on Figure \ref{fig:lb-w1}.
Assume $V(G) = \{0,1,\ldots,n-1\}$ and let
$\iota(x,y) = xn+y$ be a bijection from $V(G) \times V(G)$ to $\{0,1,\ldots,n^2-1\}$.
The gadget $G_{i,j}$ consists of $2n^2+2$ vertices
$w_{i,j}^{s,\xi}$ for $0 \leq s \leq n^2$ and $\xi \in \{1,2\}$.
For $\xi \in \{1,2\}$, $0 \leq s < n^2$ and $\iota^{-1}(s) = (x,y) \in V(G) \times V(G)$
we add an arc $(w_{i,j}^{s,\xi},w_{i,j}^{s+1,\xi})$ of weight $D$ if
$xy \in E(G)$ and $\infty$ otherwise.
Moreover, we add an arc $(w_{i,j}^{n^2,1},w_{i,j}^{0,2})$ and
arcs $(w_{i,j}^{s,1},w_{i,j}^{s,2})$ for $0 \leq s \leq n^2$, all of weight $\infty$.
We define a terminal pair 
$(w_{i,j}^{0,1},w_{i,j}^{n^2,2})$ in $G_{i,j}$.

Let us now analyze the gadget $G_{i,j}$.
The terminals $(w_{i,j}^{0,1},w_{i,j}^{n^2,2})$ are connected
by two edge-disjoint paths
$w_{i,j}^{0,1}, w_{i,j}^{1,1}, \ldots, w_{i,j}^{n^2,1}, w_{i,j}^{n^2,2}$
and
$w_{i,j}^{0,1}, w_{i,j}^{0,2}, w_{i,j}^{1,2}, \ldots, w_{i,j}^{n^2,2}$.
Therefore any solution needs to cut one edge $(w_{i,j}^{s,1},w_{i,j}^{s+1,1})$ for some $0 \leq s < n^2$
and one edge $(w_{i,j}^{s',2},w_{i,j}^{s'+1,2})$ for some $0 \leq s' < n^2$.
As the cut budget $p$ allows us to cut only $2t(t-1)$ heavy arcs (and no infinite ones),
   any solution cuts only the aforementioned two heavy arcs in each gadget $G_{i,j}$, and, apart from these, at most $t(t+1)/2$ light arcs.
Note that, moreover, we have that $s \leq s'$, as otherwise
there remains a path $w_{i,j}^{0,1}, w_{i,j}^{1,1}, \ldots, w_{i,j}^{s'+1,1}, w_{i,j}^{s'+1,2},
    w_{i,j}^{s'+2,2},\ldots,w_{i,j}^{n^2,2}$.
The index $s$ represents the choice made in gadget $G_{i,j}$, that is,
if $\iota(x,y) = s$, then we say that $G_{i,j}$ represents the pair $(x,y)$.
Note that if $G_{i,j}$ represents $s$ and $s_2 < s_1$ then $w_{i,j}^{s_2,2}$ is reachable from
$w_{i,j}^{s_1,1}$ if and only if $s_2 \leq s'$ and $s < s_1$; in particular,
  $w_{i,j}^{s,2}$ is reachable from $w_{i,j}^{s+1,1}$.

We now add connections between the gadgets to ensure the aforementioned properties,
   in a very similar fashion to \cite{marx:multicut}.

In order to ensure the first property, for every $1 \leq i < j \leq t$ and 
for every ordered pair $(x,y) \in V(G) \times V(G)$, such
that $xy \in E(G)$, we introduce: two vertices $a_{i,j}^{(x,y)}$, $b_{i,j}^{(x,y)}$,
an arc $(a_{i,j}^{(x,y)},b_{i,j}^{(x,y)})$ of weight $1$,
two arcs $(w_{i,j}^{\iota(x,y),2},a_{i,j}^{(x,y)})$,
$(w_{j,i}^{\iota(y,x),2},a_{i,j}^{(x,y)})$ of weight $\infty$
and two terminal pairs
$(w_{i,j}^{\iota(x,y)+1,1},b_{i,j}^{(x,y)})$ and
$(w_{j,i}^{\iota(y,x)+1,1},b_{i,j}^{(x,y)})$.
Observe that if $G_{i,j}$ represents $(x,y)$ then $w_{i,j}^{\iota(x,y),2}$
is reachable from $w_{i,j}^{\iota(x,y)+1,1}$ and the arc $(a_{i,j}^{(x,y)},b_{i,j}^{(x,y)})$
needs to be cut; similarly if $G_{j,i}$ represents $(y',x')$ then
$(a_{i,j}^{(x',y')},b_{i,j}^{(x',y')})$ needs to be cut.
If we are allowed to cut only one arc per choice of $1 \leq i < j \leq t$,
then $x=x'$ and $y=y'$. Thus, if we have only $\binom{t}{2}$ cuts of light arcs for the connections
introduced in this paragraph, the first property is ensured.

In order to ensure the second property, for each $1 \leq i \leq n$ and $x \in V(G)$ introduce
two vertices $c_i^x$ and $d_i^x$ connected by an arc $(c_i^x,d_i^x)$ of weight $1$.
Furthermore, for every $1 \leq j \leq n$, $j \neq i$ we add an arc
$(w_{i,j}^{\iota(x,0),2},c_i^x)$ of weight $\infty$ and a terminal
pair $(w_{i,j}^{\iota(x+1,0),1},d_i^x)$.
Note that if $G_{i,j}$ represents $(x,y)$ then $w_{i,j}^{\iota(x,0),2}$ is reachable
from $w_{i,j}^{\iota(x+1,0),1}$ and the arc $(c_i^x,d_i^x)$ needs to be cut.
If we are allowed only one cut per index $1 \leq i \leq n$ (i.e.,
$t$ cuts in total for connections introduced in this paragraph),
then the second property would be satisfied.
This concludes the description of the reduction.

To see that the constructed graph is acyclic, note that each gadget $G_{i,j}$
admits a topological order
$$w_{i,j}^{0,1}, w_{i,j}^{1,1}, \ldots, w_{i,j}^{n^2,1},w_{i,j}^{0,2}, w_{i,j}^{1,2}, \ldots, w_{i,j}^{n^2,2}.$$
Moreover, all connections between the gadgets contain outgoing edges only;
therefore a sequence that first contains all vertices of all gadgets (in
the aforementioned order within each gadget), then all pairs
$a_{i,j}^{(x,y)}, b_{i,j}^{(x,y)}$ and finally all pairs
$c_i^x, d_i^x$ is a topological order of the constructed graph.

Let us now formally prove the equivalence. Let $\{v_1,v_2,\ldots,v_t\}$ be a set
of vertices that induce a clique in $G$. Consider a set of arcs
\begin{align*}
&\{(w_{i,j}^{\iota(v_i,v_j),\xi},w_{i,j}^{\iota(v_i,v_j)+1,\xi}): 1 \leq i, j \leq t, i \neq j, 1 \leq \xi \leq 2\} \\
    &\quad \cup \{(a_{i,j}^{(v_i,v_j)},b_{i,j}^{(v_i,v_j)}: 1 \leq i < j \leq t\}\\
    &\quad \cup \{(c_i^{v_i}, d_i^{v_i}): 1 \leq i \leq t\}
\end{align*}
of weight exactly $p$. By the discussion on the gadgets $G_{i,j}$, the first group of arcs
ensures that the terminal pair in each gadget $G_{i,j}$ is separated; note that the connections between the gadgets contain only arcs outgoing from the
 gadgets $G_{i,j}$, so all the paths between considered pairs of terminals have to be entirely contained in the corresponding gadget. Moreover, for any $0 \leq s_2 < s_1 \leq n^2$,
in gadget $G_{i,j}$ the vertex $w_{i,j}^{s_2,2}$ is reachable from $w_{i,j}^{s_1,1}$
if and only if $s_2 \leq \iota(v_i,v_j) < s_1$. Therefore,
if $(x,y) \neq (v_i,v_j)$, then the first group of arcs ensure that
the terminal pair $(w_{i,j}^{\iota(x,y)+1,1}, b_{i,j}^{(x,y)})$
(or $(w_{i,j}^{\iota(x,y)+1,1}, b_{j,i}^{(y,x)})$ if $i > j$) is separated;
the second group separates the remaining pair for $(x,y) = (v_i,v_j)$.
Similarly, if $x \neq v_i$, the first group of arcs ensure that
the terminal pair $(w_{i,j}^{\iota(x+1,0),1}, d_i^x)$ is separated;
the third group separates the remaining pair for $x=v_i$.
We infer that the constructed graph admits a multicut of size $p$.

In the other direction, let $Z$ be a multicut in the constructed graph of
size at most $p$. As discussed, $Z$ needs to contain exactly two arcs
of weight $D$ from each gadget $G_{i,j}$ and each gadget $G_{i,j}$
represents some pair $(x,y)$. This leaves us with a budget of $t(t+1)/2 = \binom{t}{2} + t$
cuts of light arcs. We infer that we can spend one cut of
an arc $(a_{i,j}^{(x,y)},b_{i,j}^{(x,y)})$ per a pair $1 \leq i < j \leq t$
and only one cut of an arc $(c_i^x,d_i^x)$ per an index $1 \leq i \leq t$.
Therefore, both properties of what the gadgets $G_{i,j}$ may represent
are satisfied, so there are distinct vertices $v_1,v_2,\ldots,v_t$ such that
gadget $G_{i,j}$ represents $(v_i,v_j)$. As in each gadget a finite weight
was assigned only to an arc that corresponds to an edge in $G$,
    we obtain a clique of size $t$ in $G$.
\end{proof}

\subsection{NP-hardness of \dagmulticut{} with constant number of terminals}

We start with a proof of Theorem \ref{thm:lb-skew}. Then, we derive Theorem \ref{thm:lb-np} from Theorem \ref{thm:lb-skew}.

\begin{figure}[htp]
\begin{center}
\includegraphics[scale=1]{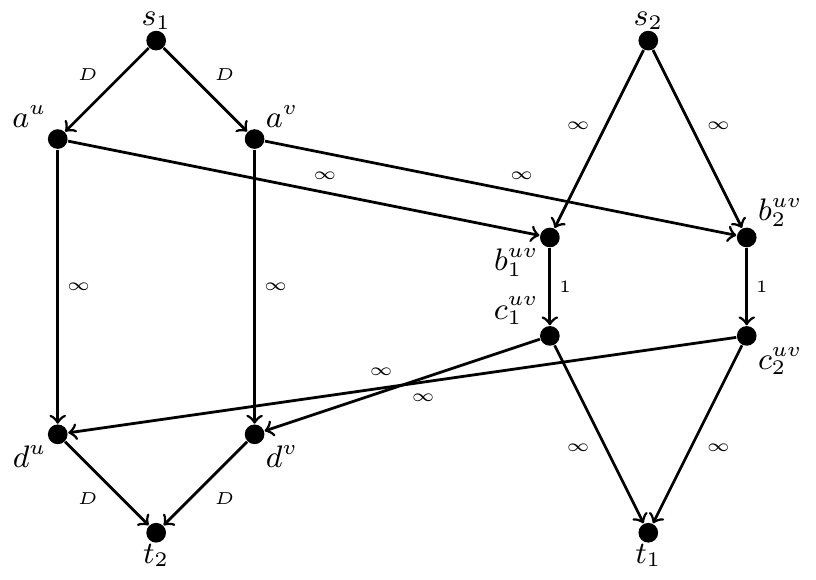}
\caption{A part of the construction in the proof
 of Theorem \ref{thm:lb-skew} with paths $P_u$ and $P_v$ and their connection
 due to an edge $uv$.}
\label{fig:lb-skew}
\end{center}
\end{figure}

\begin{proof}[Proof of Theorem \ref{thm:lb-skew}]
We provide a reduction from the NP-complete \textsc{Max-Cut} problem. Let us remind, that the \textsc{Max-Cut} instance $(G,t)$ is an undirected graph together with an integer $t$ and we ask
for a subset of vertices $X \subseteq V(G)$ such that
there are at least $t$ edges of $G$ with exactly one endpoint in $X$.
Denote $|V(G)| = n$ and $|E(G)| = m$.

We construct an equivalent \dagmulticut{} instance this time in the arc-deletion setting; let us remark that the arc- and vertex-deletion variants are equivalent (cf. \cite{dir-mwc}).
For clarity, we allow arcs to have weights: in our construction, we
use infinite (of weight $p+1$, denoted $\infty$; $p$, the budget for cuts, will be polynomial in the size of $G$), heavy (of weight $D = 2m+1$) and light (of weight $1$) arcs.
As the weights are polynomial in the size of $G$, we can easily reduce
the weighted variant to the unweighted one by replacing an arc $uv$ of weight $\omega$ with $\omega$ $uv$-paths of length two.

We start a construction of an equivalent \textsc{Skew Multicut} instance
by setting the cut budget $p = nD + 2m-t$ (as $p < nD+D$, we can delete
    only $n$ heavy edges) and
by introducing two sources $s_1$, $s_2$ and two sinks $t_1$, $t_2$;
recall that the set of terminal pairs is defined as $\terms = \{(s_1,t_1), (s_1,t_2), (s_2,t_2)\}$.

For each vertex $v \in V(G)$ we introduce two vertices $a^v$ and $d^v$, as well as two arcs $(s_1,a^v)$
and $(d^v, t_2)$ of weight $D$ and an arc $(a^v,d^v)$ of weight $\infty$. We denote the path $s_1,a^v,d^v,t_2$
as $P_v$; note that any solution needs to cut one of the heavy arcs $(s_1,a^v)$ or $(d^v,t_2)$.
As $p < nD+n$, each path $P_v$ is cut exactly once and the choice of the cut arc corresponds to the choice
whether $v \in X$ or $v \in V(G) \setminus X$.

We now connect the paths $P_v$ in such a way that for an edge $uv \in E(G)$ we profit
if the paths $P_u$ and $P_v$ are cut in a different manner.
For each edge $uv \in E(G)$ we introduce four vertices
$b_\alpha^{uv}$, $c_\alpha^{uv}$ for $\alpha \in {1,2}$,
two arcs
$(b_1^{uv},c_1^{uv})$ and $(b_2^{uv},c_2^{uv})$ of weight $1$,
  and eight arcs: $(s_2,b_\alpha^{uv})$, $(c_\alpha^{uv},t_1)$ for $\alpha \in {1,2}$
  as well as
$(a^u,b_1^{uv})$, $(a^v,b_2^{uv})$,
$(c_2^{uv},d^u)$, $(c_1^{uv},d^v)$ of weight $\infty$.
Note that the construction is symmetric with regard to $u$ and $v$
(i.e., changing the names of $u$ and $v$ results in changing the names of $b_1^{uv}$ and $c_1^{uv}$ with $b_2^{uv}$ and $c_2^{uv}$).
The intuition behind this construction is as follows: if the paths $P_u$ and $P_v$ are cut in a different manner,
we need to cut only one arc of weight $1$ for the edge $uv$, and otherwise we need to cut both arcs.
Part of the construction, with paths $P_u$, $P_v$ and the connection corresponding to the edge $uv$,
is depicted on Figure \ref{fig:lb-skew}.

The following topological order of the constructed graph proves that
the we indeed construct an acyclic graph
(within each set, we order the vertices arbitrarily):
\begin{align*}
\big\langle s_1,s_2,&\{a^v: v \in V(G)\},\{b_\alpha^{uv}: 1 \leq \alpha \leq 2, uv \in E(G)\},\\
    &\{c_\alpha^{uv}: 1 \leq \alpha \leq 2, uv \in E(G)\}, \{d^v: v \in V(G)\}, t_1, t_2 \big\rangle
\end{align*}

Let us now formally prove the equivalence of the input and output instances.
Let $X \subseteq V(G)$ be such that there are at most $m-t$ edges in $E(G[X]) \cup E(G \setminus X)$.
Consider the following set
\begin{align*}
Z &= \{(s_1,a^v): v \in X\} \cup \{(d^v,t_2): v \in V(G) \setminus X\} \\
  &\quad \cup \{(b_1^{uv},c_1^{uv}): u \in V(G) \setminus X \vee v \in X\}\\
  &\quad \cup \{(b_2^{uv},c_2^{uv}): u \in X \vee v \in V(G) \setminus X\}.
\end{align*}
Intuitively, if $v\in X$ then we take the arc $(s_1,a^v)$ to the solution, and otherwise we take the arc $(d^v,t_2)$. If $u\in X$ and $v\in V(G)\setminus X$ then we only need to include the arc $(b_2^{uv},c_2^{uv})$ in the solution; similarly, if $u\in V(G)\setminus X$ and $v\in X$ then we only need to include the arc $(b_1^{uv},c_1^{uv})$. However, if $u\in X$ and $v\in X$, or $u\in V(G)\setminus X$ and $v\in V(G)\setminus X$, then both of the arcs $(b_\alpha^{uv},c_\alpha^{uv})$ for $\alpha \in \{1,2\}$ need to be taken. As at least $t$ edges in $G$ have exactly one endpoint in $X$, we infer that the weight of $Z$ is at most $p$. It remains to check that $Z$ is a multicut in the constructed \textsc{Skew Multicut} instance.

First, consider the terminal $s_1$. Its out-arc $(s_1,a^u)$ is not in $Z$ iff $u \in V(G) \setminus X$.
The out-neighbours of $a^u$ are $d^u$ and $b_1^{uv}$, $b_2^{wu}$ for all edges $uv, wu \in E(G)$.
From the construction of $Z$ we infer that $(d^u,t_2) \in Z$ and $(b_1^{uv},c_1^{uv}), (b_2^{wu},c_2^{wu}) \in Z$, thus
$Z$ is a $s_1-\{t_1,t_2\}$ separator. Symmetrically we show that $Z$ is a $\{s_1,s_2\}-t_2$
separator, and the constructed instance is a YES-instance to \textsc{Skew Multicut}.

In the other direction, let $Z$ be a solution to the constructed instance of weight at most $p$.
As discussed, $Z$ needs to contain exactly one heavy arc for each $v \in V(G)$: $(s_1,a^v)$ or $(d^v,t_2)$,
   and we are left with a budget of at most $2m-t$ light arcs.
Let $X \subseteq V(G)$ be defined as the set of those vertices $v \in V(G)$ for which $(s_1,a^v) \in Z$.

Consider an edge $uv \in E(G)$.
If $u \in X$ then $Z$ needs to intersect
the path $s_2-b_2^{uv}-c_2^{uv}-d^u-t_2$ in $(b_2^{uv},c_2^{uv})$,
and if $u \in V(G) \setminus X$ then $Z$ needs to intersect 
the path $s_1-a^u-b_1^{uv}-c_1^{uv}-t_1$ in $(b_1^{uv},c_1^{uv})$.
Similarly, if $v \in X$ then $Z$ needs to intersect
the path $s_2,b_1^{uv},c_1^{uv},d^v,t_2$ in $(b_1^{uv},c_1^{uv})$,
and if $v \in V(G) \setminus X$ then $Z$ needs to intersect
the path $s_1,a^v,b_2^{uv},c_2^{uv},t_1$ in $(b_2^{uv},c_2^{uv})$.
We infer that for each edge $uv \in E(G)$, $Z$ needs to contain at least
one arc $(b_\alpha^{uv},c_\alpha^{uv})$ for $\alpha \in \{1,2\}$,
and both of them if $u,v \in X$ or $u,v \in V(G) \setminus X$.
As $Z$ contains at most $2m-t$ light edges, $X$ is a solution
to the \textsc{Max-Cut} instance $(G,t)$.
\end{proof}

Theorem \ref{thm:lb-np} follows from an easy reduction from \textsc{Skew Multicut}.

\begin{proof}[Proof of Theorem \ref{thm:lb-np}]
For clarity,
in this proof we consider arc-deletion variant, similarly as in the proof of Theorem \ref{thm:lb-np}; again, let us remark that the arc- and vertex-deletion variants are equivalent (cf. \cite{dir-mwc}).
Let $(G,\terms,p)$ be an instance \textsc{Skew Multicut} with two sinks and two sources,
(i.e., $\terms = \{(s_1,t_1), (s_1,t_2), (s_2,t_2)\}$), whose NP-completeness
is established by Theorem \ref{thm:lb-skew}.
Moreover, we may assume that the terminals $s_1$, $s_2$, $t_1$, $t_2$ are pairwise
distinct and that $\inN_G(s_1) = \inN_G(s_2) = \outN_G(t_1) = \outN_G(t_2) = 0$:
we can observe that the graph constructed in the proof of Theorem \ref{thm:lb-np}
satisfies these conditions, or apply a reduction in the flavour of the proof of Lemma \ref{lem:source-sink}
to the instance $(G,\terms,p)$.

Let $G'$ be constructed from $G$ by adding an arc $(t_2,t_1)$ of infinite weight
(or, equivalently, $p+1$ $t_2t_1$-paths of length two) and let
$\terms' = \{(s_1,t_1), (s_2,t_2)\}$. We claim that a \dagmulticut{} instance
$(G',\terms',p)$ is equivalent to the \textsc{Skew Multicut} instance $(G,\terms,p)$.

In one direction, let $Z$ be a solution to $(G,\terms,p)$. If $P$ is a $s_2t_2$-path
in $G'$ then $P$ is also present in $G$ and $Z$ intersects $P$. If $P$ is a $s_1t_1$-path
in $G'$ then either $P$ is present in $G$ or $P$ ends with the arc $(t_2,t_1)$; in the second
case $P$ contains a $s_1t_2$-path in $G$ and in both cases $Z$ intersects $P$.
We infer that $Z$ is a solution to $(G',\terms',p)$ as well.

In the other direction, let $Z$ be a solution to $(G',\terms',p)$.
Note that $(t_2,t_1) \notin Z$, as $(t_2,t_1)$ has infinite weight
(or, equivalently, at least one of the $t_2t_1$-paths is not intersected by $Z$).
For a $s_1t_1$- or $s_2t_2$-path $P$ in $G$, $P$ is also present in $G'$
and $Z$ intersects $P$. If $P$ is a $s_1t_2$-path in $G$
we can prolong $P$ with the arc $(t_2,t_1)$ in $G'$ to obtain a $s_1t_1$-path $P'$.
As $Z$ is a $s_1-t_1$ separator in $G'$, we infer that
$Z$ intersects $P'$ and therefore $P$ as well. 

\end{proof}

%% file: conclusions.tex
\section{Conclusions}\label{sec:conc}

The results of this paper unravel the full picture of the parameterized complexity of \dagmulticut{}.
A natural follow-up question is the complexity of \multicut{} in general directed graphs, where we so far know only that the case of two terminal pairs is FPT \cite{dir-mwc}
and the cutset parameterization is W[1]-hard~\cite{marx:multicut}.
The assumption of acyclicity seems to be crucial for our approach
in Lemma~\ref{lem:magic} and subsequent Branching \ref{branch:magic}.
We also note that, although
an existence of a polynomial kernelization algorithm
for most graph separation problems in directed graphs was recently refuted \cite{no-colours-ids},
the question of a polynomial kernel for \dagmulticut{} remains open.